\documentclass[11pt]{article}
\usepackage{fullpage}
\usepackage{authblk}
\usepackage{amsmath,amssymb,amsfonts,amsthm}
\usepackage{bm,bbm}
\usepackage{array}

\usepackage{algorithm}
\usepackage{algpseudocode}

\usepackage{graphicx}
\usepackage{subfig}
\usepackage{float}
\usepackage{placeins}
\usepackage[dvipsnames]{xcolor}
\usepackage{pifont}
\usepackage{todonotes}
\usepackage{tikz}
\usetikzlibrary{calc,patterns,decorations.pathmorphing,decorations.markings}
\usepackage{comment}
\usepackage[normalem]{ulem}
\usepackage[english]{babel}
\usepackage{url}
\usepackage[numbers,sort&compress]{natbib}
\usepackage[colorlinks=true]{hyperref}
\usepackage{wrapfig}
\usepackage{needspace}

\newtheorem{definition}{Definition}

\newcommand{\change}[1]{\ensuremath{\operatorname{#1}}}
\newcommand{\MAT}{ [ \begin{array}}  
\newcommand{\mat}{\end{array}  ]}

\newtheorem{Lemma}{Lemma}[section]
\newtheorem{Theorem}{Theorem}[section]

\newtheorem{Q}{Question}[section]

\def\Pr{{\mathbb P}}
\def\CC{{\mathbb C}}
\def\RR{{\mathbb R}}

\newcommand{\argmin}{\operatorname{argmin}}

\newcommand{\vct}[1]{\mathbf{#1}}
\newcommand{\tTx}[1]{\mathbf{#1}}

\newcommand{\bmtx}[1]{\mathbf{#1}}

\newcommand{\<}{\langle}
\renewcommand{\>}{\rangle}

\newcommand{\calA}{\mathcal{A}}

\newcommand{\calM}{\mathcal{M}}

\newcommand{\calP}{\mathcal{P}}

\newcommand{\tT}{\mathcal{T}}

\newcommand{\va}{\vct{a}}
\newcommand{\vb}{\vct{b}}
\newcommand{\vc}{\vct{c}}

\newcommand{\ve}{\vct{e}}
\newcommand{\vf}{\vct{f}}
\newcommand{\vg}{\vct{g}}
\newcommand{\vh}{\vct{h}}

\newcommand{\vp}{\vct{p}}
\newcommand{\vq}{\vct{q}}

\newcommand{\vs}{\vct{s}}

\newcommand{\vu}{\vct{u}}
\newcommand{\vv}{\vct{v}}

\newcommand{\vx}{\vct{x}}
\newcommand{\vy}{\vct{y}}
\newcommand{\vz}{\vct{z}}
\newcommand{\vzero}{\vct{0}}

\newcommand{\mB}{\tTx{B}}

\newcommand{\mF}{\tTx{F}}

\newcommand{\mP}{\tTx{P}}

\newcommand{\mX}{\tTx{X}}

\newcommand{\mId}{\bmtx{I}}

\newcommand{\tensor}[1]{\boldsymbol{\mathcal{#1}}}
\def \tA {\tensor{A}}

\def \tE {\tensor{E}}

\def \tT {\tensor{T}}

\def \tX {\tensor{X}}
\def \tY {\tensor{Y}}
\def \tZ {\tensor{Z}}

\begin{document}

\title{Landscape Analysis of Simultaneous Blind Deconvolution and Phase Retrieval via Structured Low-Rank Tensor Recovery}
\author[1]{Xiao Liang\textsuperscript{*,$\dagger$}}
\author[2]{Zhen Qin\textsuperscript{$\dagger$}}
\author[3]{Zhihui Zhu}
\author[1]{Shuang Li}
\affil[1]{Department of Electrical and Computer Engineering, Iowa State University, Ames, Iowa 50011, USA}
\affil[2]{Michigan Institute for Computational Discovery and Engineering, Department of Electrical Engineering and Computer Science, Department of Statistics, University of Michigan, Ann Arbor, Michigan 48109, USA}
\affil[3]{Department of Computer Science and Engineering, The Ohio State University, Columbus, Ohio 43210, USA}
\date{}
\maketitle
\begingroup
\renewcommand{\thefootnote}{\fnsymbol{footnote}}
\footnotetext[1]{Corresponding author. Email addresses:
\texttt{liangx@iastate.edu} (Xiao Liang),
\texttt{zhenqin@umich.edu} (Zhen Qin),
\texttt{zhu.3440@osu.edu} (Zhihui Zhu), and
\texttt{lishuang@iastate.edu} (Shuang Li).}
\footnotetext[2]{Xiao Liang and Zhen Qin contributed equally to this work.}
\endgroup

\begin{abstract}
This paper studies the landscape of the simultaneous blind deconvolution and phase retrieval (BDPR) problem via a structured low-rank tensor recovery framework. 
Since the complicated sensing structure makes direct landscape analysis technically challenging, we introduce a tractable tensor sensing surrogate that preserves the essential low-rank structure.
We first analyze its population risk, characterize the global landscape on the unit sphere, and establish local linear convergence of Riemannian gradient descent (RGD) under mild conditions.
We then extend the analysis to tensor sensing under a tensor restricted isometry property (TRIP), establishing local geometric and RGD convergence guarantees, and quantify robustness to measurement noise.
Our theoretical results are further supported by extensive numerical experiments. 
These findings offer insights into the optimization landscape of the surrogate tensor sensing problem, thereby providing guidance for studying and solving the original BDPR problem.
\end{abstract}

\noindent\textbf{Keywords:} Blind deconvolution, Phase retrieval, Tensor factorization, Tensor sensing, Geometric landscape

\section{Introduction}
Blind deconvolution and phase retrieval are fundamental inverse problems in signal processing~\cite{BD_signal, PR_signal}, machine learning~\cite{BD_ML, PR_ML}, and computational imaging~\cite{BD_com, PR_com}. Blind deconvolution aims to simultaneously recover an unknown signal and an unknown convolution kernel from their convolved measurements~\cite{BD_Krishnan, BD_Levin, BD_Jin_2018_ECCV, 5206815}. In imaging, this corresponds to reconstructing a sharp image from a blurred observation without knowing the blur kernel. Meanwhile, phase retrieval focuses on recovering a complex-valued signal from the magnitudes of its linear measurements, with phase information missing~\cite{PR_Shechtman, PR_Fannjiang, 1142559, candes2015phase}.

Blind deconvolution and phase retrieval are both 
ill-posed and nonconvex, and often arise jointly in applications such as optical imaging and communications. This has motivated increasing efforts to jointly model and solve the two problems.
For instance, Shamshad et al.~\cite{shamshad2020class} proposed alternating gradient descent using pretrained generative priors.
Ahmed et al.~\cite{BDPR_JMLR:v20:19-332} introduced 
a convex lifting approach for phaseless Fourier measurements in random subspaces.
Fu et al.~\cite{BDPR_Fu} addressed simultaneous 
BDPR in optical wireless communications using low-rank matrix lifting and difference-of-convex programming, demonstrating improved recovery and noise robustness.

Notably, Li et al.~\cite{li2019simultaneous} studied a BDPR problem in phase imaging and
reformulated it as low-rank tensor recovery, solved via iterative hard thresholding.
However, the method lacks performance and convergence guarantees due to the intricate sensing structure; its optimization landscape and robustness to measurement noise also remain unexplored.
These limitations motivate our analysis of the BDPR landscape via a tractable surrogate.

Recent work has advanced landscape analysis for blind deconvolution and phase retrieval separately.
Zhang et al.~\cite{zhangstructured} formulated blind deconvolution as a nonconvex optimization problem and  showed that, within a certain region of the sphere, local optima lie near shifted truncations of the ground truth.
Li et al.~\cite{li_global2018} established that all local minima correspond to the true inverse filter and all saddle points are strict.
D\'iaz~\cite{diaz2019nonsmooth} characterized the landscape of a nonsmooth blind deconvolution objective and numerically demonstrated recovery from random initialization.
For phase retrieval, Chen and Cand\`es~\cite{chen2015solving} proved exact recovery using spectral initialization and adaptive gradient descent, while~\cite{sun2018geometric} established benign geometry of the nonconvex least-squares formulation of phase retrieval under generic measurements.

The landscape analysis of tensor decomposition and tensor sensing problems has also been extensively investigated.
For random overcomplete tensors, 
Ge and Ma~\cite{ge2017optimization} showed that local maxima above a certain objective-value threshold are approximately globally optimal.
For exact Tucker tensors, Frandsen and Ge~\cite{frandsen2022optimization} established that all local minima of an appropriately regularized objective are globally optimal.
Li and Li~\cite{li2022local} analyzed Burer--Monteiro factorization for well-conditioned convex objectives, deriving local convergence guarantees and characterizing global rank-one approximation landscapes for orthogonally decomposable tensors.
Kileel et al.~\cite{kileel2021landscape} showed that second-order critical points above a threshold correspond to true tensor components in the noiseless case and approximate them in the noisy case.
For tensor-train sensing, Qin et al.~\cite{qin2024guaranteed} proved linear convergence of RGD under RIP and later extended it to broader structured recovery models~\cite{qin2025scalable}.

In this work, we analyze the geometric landscape of the BDPR problem studied in~\cite{li2019simultaneous}. 
Although BDPR can be reformulated as structured low-rank tensor recovery, the intricate sensing structure makes direct landscape analysis intractable. 
To overcome this difficulty, we introduce the tensor sensing problem~\eqref{loss_sensing} as a tractable surrogate that preserves the essential low-rank structure, and first study its population risk---formulated as a tensor factorization problem~\eqref{The loss function of normalization CP PCA real}---which captures key aspects of the unknown low-rank tensor. 
Although the landscape of tensor factorization and tensor sensing problems has been extensively studied, existing landscape results do not directly apply because our tensors have the partially symmetric form $\vx \circ \vx \circ \vh$, with two modes sharing $\vx$, imposing constraints absent from general CP models and altering the loss geometry.  Analyzing this surrogate provides insights into the original BDPR landscape and guidance for algorithm design.

Our main contributions are summarized as follows.
\begin{itemize}

    \item To gain insights into the optimization landscape of the BDPR problem, we first analyze the global geometry of the tensor sensing surrogate’s population risk on the unit sphere, characterize all critical points, and establish local linear convergence of RGD near the ground truth.
    
    \item We then extend the analysis to the tensor sensing surrogate, proving local linear convergence of RGD under TRIP conditions and with suitable initialization.
    
    
    \item We further extend the analysis to noisy measurements, deriving explicit error bounds and showing that RGD 
converges linearly to a neighborhood of the ground truth whose radius scales with the noise level.    
    
    \item Extensive numerical experiments validate the predicted convergence behavior, robustness to noise, and effectiveness of the initialization.
\end{itemize}



{\bf Notation:} 
We denote vectors by bold lowercase letters (e.g., $\vx$), matrices by bold uppercase letters (e.g., $\mX$), and tensors by bold calligraphic letters (e.g., $\tA$). 
Their entries are written as $\vx(i)$, $\mX(i,j)$, and $\tA(n_1,\dots,n_D)$, while $\mX(:,j)$ and $\mX(i,:)$ denote matrix columns and rows. 
We use $(\cdot)^\top$, $(\cdot)^H$, and $(\cdot)^*$ for transpose, Hermitian (conjugate transpose), and complex conjugate, respectively, and $\|\cdot\|_F$ for Frobenius norm.
The symbols $\circ$, $\otimes$, $\odot$, and $\circledast$ denote the outer product, Kronecker product, Hadamard product, and circular convolution, respectively. 
For a tensor $\mathcal{\tX}$, $\mathcal{M}_n(\mathcal{\tX})$ denotes its mode-$n$ matricization (unfolding).


\section{Preliminaries}
\label{sec:Preliminaries}

A $D$-th order tensor is an array in $\mathbb{C}^{N_1\times\cdots\times N_D}$. Vectors and matrices correspond to special cases with $D=1$ and $D=2$, respectively. In this work, we focus on third-order tensors ($D=3$).
For $\tX,\tY\in\mathbb{C}^{N_1\times N_2\times N_3}$, define the inner product as
$
\langle \tX,\tY\rangle
=
\sum_{n_1=1}^{N_1}\sum_{n_2=1}^{N_2}\sum_{n_3=1}^{N_3}
\tX(n_1,n_2,n_3)\tY(n_1,n_2,n_3)^*,
$
and the induced Frobenius norm is
$
\|\tX\|_F=\sqrt{\langle\tX,\tX\rangle}.
$

A third-order tensor $\tX \in \mathbb{C}^{N_1 \times N_2 \times N_3}$ is defined to be of {\em rank one} if it admits a decomposition of the form
$\tX = \va \circ \vb \circ \vc,$
where $\va \in \mathbb{C}^{N_1}$, $\vb \in \mathbb{C}^{N_2}$, and $\vc \in \mathbb{C}^{N_3}$ are nonzero vectors. The notation $\circ$ denotes the outer product.

Common tensor decompositions include CANDECOMP/PARAFAC (CP) decomposition~\cite{kolda2009tensor}, Tucker decomposition~\cite{tucker1966some}, and tensor train decomposition~\cite{oseledets2011tensor}. 
In this work, we focus on CP decomposition, which represents a third-order tensor as a linear combination of rank-one tensors, where each rank-one tensor is formed by the outer product of three vectors:
$
\tX=\sum_{i=1}^r \lambda_i \va_i\circ\vb_i\circ\vc_i,
$
where $\lambda_i \in \mathbb{R}$ are scalar weights, $\va_i \in \mathbb{C}^{N_1}$, $\vb_i \in \mathbb{C}^{N_2}$, and $\vc_i \in \mathbb{C}^{N_3}$ are vectors.
The smallest integer $r$ for which this decomposition holds is the CP rank of $\tX$.

A fiber is the higher-order analogue of a matrix row or column and is obtained by fixing all but one tensor index, the unfixed index determines its mode. For a third-order tensor, mode-1, mode-2, and mode-3 fibers are referred to as column, row, and tube fibers, respectively. Mode-$n$ matricization $\mathcal{M}_n(\tX)$ arranges all mode-$n$ fibers as matrix columns. 
For a third-order rank-one tensor $\mathbf{a} \circ \mathbf{b} \circ \mathbf{c}\in \mathbb{C}^{N_1 \times N_2 \times N_3}$, the mode-$n$ unfoldings can be expressed as:
    $\calM_1(\va \circ \vb \circ \vc) = \va (\vc\otimes \vb)^\top\in\CC^{N_1\times N_2N_3}$,
    $\calM_2(\va \circ \vb \circ \vc) = \vb (\va\otimes\vc)^\top\in\CC^{N_2\times N_3N_1},$
    $\calM_3(\va \circ \vb \circ \vc) = \vc (\vb\otimes \va)^\top\in\CC^{N_3\times N_2N_1}.$

\section{Problem Formulation}
\label{sec:problem_formulation}
The BDPR problem investigated in this work arises in phase imaging from a defocused intensity stack, where intensity-only measurements are captured by the detector placed at multiple positions along the optical axis. 
Under coherent illumination, the intensity measurement at the $i$-th detector position is
$\vy_i=|\mF(\vx\odot\vg_i)|^2$, $i=1,\ldots,I$, where
$\mF\in\CC^{N\times N}$ is the discrete Fourier transform (DFT)  matrix with $\mF(n_1,n_2)=e^{-j2\pi(n_1-1)(n_2-1)/N}$, $\vg_i\in\CC^N$ is the corresponding Gaussian chirp phase mask, $\vx\in\CC^N$ is the unknown sample field, and $I$ is the number of detector positions.
We instead consider partially coherent illumination from sources such as LEDs, incandescent bulbs, and X-ray tubes. Compared with traditional coherent illumination, it offers higher light throughput, reduced speckle, and improved depth sectioning~\cite{barone2002quantitative,pfeiffer2006phase}.

By the Van Cittert--Zernike theorem~\cite{born2013principles}, partial coherence can be modeled through the source shape, yielding a partially coherent forward model with an extra convolution~\cite{li2019simultaneous}:

\begin{equation}
    \vy_i = |\mF(\vx \odot \vg_i)|^2 \circledast (\mP_i \vs),\quad i = 1,\ldots,I,
    \label{BD_PR_NL}
\end{equation}
where $\vs \in \RR^{N}$ is the source shape vector and represents the unknown discretized source distribution function, and $\mP_i \in \RR^{N \times N}$ is a known linear operator that scales the source shape according to the $i$-th detector position.

Following~\cite{li2019simultaneous}, we assume that the unknown source shape $\vs$ lies in a known low-dimensional subspace, $\vs=\mB\vh$, where $\mB\in\RR^{N\times K}$ is a known basis matrix, $K\ll N$, and $\vh\in\RR^K$ is an unknown coefficient vector. Thus, estimating $\vs$ reduces to estimating $\vh$. In this phase imaging model, the goal is to jointly recover the complex field of the sample $\vx$ and source shape $\vs$ (equivalently, $\vh$) from the intensity measurements in~\eqref{BD_PR_NL}. This task involves both phase retrieval from intensity-only data and blind deconvolution due to the unknown source shape.

To address this problem, the authors in~\cite{li2019simultaneous} reformulate it as a structured low-rank tensor recovery problem by rewriting the $n$-th entry of $\vy_i$:
\begin{align}
 \vy_i(n)  = \left\langle \vx^* \circ \vx \circ \vh,\tA_{i,n} \right\rangle,
 \label{eqn:measurement}
\end{align} 
where $\tA_{i,n} \in \CC^{N\times N \times K}$ is a structured sensing tensor that encodes the measurement process with the $(n_1,n_2,k)$-th entry defined as
\begin{equation}
\begin{aligned}
\tA_{i,n} ( n_1, n_2, k) =& \frac{1}{N} \!  
\left[ \mF(:,n_1) ^\top   \vg_i(n_1) \right]   \!\odot   \!
\left[ \mF(:,n_2)^H \vg_i(n_2)^*  \right] \\
&  \mF^H \change{diag}(\mF(:,n)) \mF^* \mP_i^* \mB(:,k)^*.
\end{aligned}
\label{def_senT}
\end{equation}
Then, the original nonlinear model is transformed into a linear observation model over a rank-one third-order tensor formed by the outer product $\vx^* \circ \vx \circ \vh$. 
This reformulation enables tensor recovery techniques to jointly estimate the complex-valued sample $\vx$ and the source coefficient vector $\vh$.

To solve the resulting structured low-rank tensor recovery problem, Li et al.~\cite{li2019simultaneous} proposed an algorithm based on Tensor Iterative Hard Thresholding, exploiting the rank-one structure for efficient recovery. 
Despite promising empirical performance, the work lacks theoretical guarantees: neither its convergence nor the geometric landscape is analyzed. 
Consequently, important theoretical questions concerning the optimization landscape—such as the existence of spurious local minima and the global geometry of the problem—remain unexplored, which motivates the analysis in this work.

\section{Tensor Factorization}
\label{sec: tensor factorization}

Motivated by \eqref{eqn:measurement}, in this and subsequent sections, we aim to provide a thorough landscape analysis and develop efficient nonconvex optimization methods with guaranteed convergence for solving the following factorized rank-one partially symmetric tensor recovery problem:
\begin{align}
    \label{General optimization problem-1}
    \min_{\substack{
    \vx \in \mathbb{R}^N,\, \|\vx\|_2 = 1 \\
    \vh \in \mathbb{R}^K
}}
f(\vx,\vh) = \frac{1}{2}\|\calA(\vx \circ \vx \circ \vh) - \vy\|_2^2,
\end{align}
where $\calA: \RR^{N\times N \times K} \to \RR^m$ is a  linear sensing operator with
\begin{eqnarray}
    \label{Definition of tensor sensing}
    \vy=\calA(\tT^\star) = \begin{bmatrix}
\<\tA_1, \tT^\star \> ~ \cdots ~ \<\tA_m, \tT^\star \>
    \end{bmatrix}^\top\in\RR^m,
\end{eqnarray}
and $\tT^\star= \vx^\star \circ \vx^\star \circ \vh^\star \in \RR^{N\times N \times K}$ denotes the ground-truth low-rank tensor. To remove the inherent scalar ambiguity among $\vx$ and $\vh$ (i.e., $\alpha\vx \circ \alpha\vx \circ \tfrac{1}{\alpha^2}\vh = \vx \circ \vx \circ \vh$ for any $\alpha\neq 0$), here we impose the normalization constraint $\|\vx\|_2 = 1$.
To simplify the presentation, in particular, to avoid introducing Wirtinger derivatives for complex variables, we focus on the real-valued setting for clarity of exposition. The analysis naturally extends to the complex domain.

To guarantee recovery from limited linear measurements, the sensing operator $\calA$ must satisfy certain properties. We will present one such property and study the corresponding factorized problem in the next section. 
In this section, we first study the case where the sensing operator \(\mathcal A\) is the identity map. The resulting objective corresponds to the population risk~\cite{li2019landscape,li2022landscape} of~\eqref{General optimization problem-1} under isotropic random sensing. The problem then reduces to the following tensor factorization problem:
\begin{eqnarray}
    \label{The loss function of normalization CP PCA real}
    \min_{\substack{
    \vx \in \mathbb{R}^N,\, \|\vx\|_2 = 1 \\
    \vh \in \mathbb{R}^K
}}
f(\vx,\vh) = \frac{1}{2}\|\vx \circ \vx \circ \vh - \tT^\star\|_F^2.
\end{eqnarray}
This surrogate formulation allows us to isolate and analyze the fundamental geometric properties of the objective, which will guide the subsequent analysis for the general sensing setting.

Based on this formulation, we first characterize the global landscape of problem~\eqref{The loss function of normalization CP PCA real} by identifying all its critical points (Section~\ref{sub:global_fac}), and then analyze the local convergence properties of Riemannian gradient descent (RGD), establishing linear convergence guarantees (Section~\ref{sub:local_fac}).

\subsection{Global Geometry}
\label{sub:global_fac}
Although the landscape of matrix factorization and general tensor factorization has been extensively studied, their results cannot be directly applied to our problem. The BDPR formulation induces a special symmetric CP structure of the form $\vx \circ \vx \circ \vh$, which differs fundamentally from both matrix factorization (bilinear in two factors) and general tensor factorization. This distinctive structure leads to high-order saddle points and allows us to derive tighter geometric results (e.g., global landscape characterization and explicit local convergence rates for RGD) that are not captured by existing analyses.

To investigate the global landscape of problem~\eqref{The loss function of normalization CP PCA real}, we first characterize its critical points by deriving the corresponding Euclidean and Riemannian gradients. In particular, the Euclidean gradients of $f(\vx,\vh)$ with respect to $\vx$ and $\vh$ are
\begin{align*}
    \nabla_{\vx}f(\vx,\vh) &= \!( \calM_1( \tT \!-\! \tT^\star)) (\vh\!\otimes\! \vx) \!+\! (\calM_2(\tT\! -\!\tT^\star) ) (\vx\!\otimes\!\vh)\\&= 2 \|\vh \|_2^2\|\vx\|_2^2\vx - 2\<\vh^\star,\vh  \>\<\vx^\star, \vx \>\vx^\star,\\
    \nabla_{\vh}f(\vx,\vh) &=  (\calM_3(\tT - \tT^\star))(\vx\otimes \vx) = \|\vx\|_2^4\vh - \< \vx^\star,\vx \>^2\vh^\star,
\end{align*}
where $\tT = \vx \circ \vx \circ \vh$ and $\tT^\star = \vx^\star \circ \vx^\star \circ \vh^\star$. 

To obtain the Riemannian gradient on the unit sphere, we project the Euclidean gradient $\nabla_{\vx}f(\vx,\vh)$ onto the tangent space at $\vx$, defined by
$\text{T}_{\vx} \text{St}:=\{\vz\in\RR^N: \vx^\top\vz= 0 \},$ 
using the projection operator 
$\calP_{\text{T}_{\vx} \text{St}}(\vy ) = (\mId - \vx\vx^\top)\vy.$
Then, we get the Riemannian gradient of $f(\vx,\vh)$ with respect to $\vx$:
$$\text{grad}_{\vx}f(\vx,\vh) = \calP_{\text{T}_{\vx} \text{St}}(\nabla_{\vx}f(\vx,\vh)) 
    = - 2\<\vh^\star,\vh  \>\<\vx^\star, \vx \>\vx^\star   + 2\<\vh^\star,\vh  \>\<\vx^\star, \vx \>^2\vx.
$$

By solving the first-order optimality conditions, we can verify that the following two cases exhaust all possible critical points: 
    (1) $\vx\perp \vx^\star$, $\vh ={\bf 0}$;
    (2) $\vx = \pm\vx^\star$, $\vh = \vh^\star$.
With some straightforward calculations (see the Appendix~\ref{Proof of Hessian matrix in the appendix} for details), the bilinear form of the (hybrid) Riemannian Hessian at a critical point $(\vx,\vh)$ takes the form:
\begin{eqnarray}
\begin{aligned}
    \label{Riemannian Hessian total main paper}
    \text{Hess} f\!(\vx,\!\vh)\![\va,\va] \!=\! 2\|\vh\|_2^2\|\va_1\!\|_2^2 \!- \!2\langle\vh^\star\!,\vh\rangle\!\langle\vx^\star\!,\va_1\!\rangle^2 \!+\! \|\va_2\!\|_2^2 \!- \!4\langle\vh^\star\!,\va_2\!\rangle\!\langle\vx^\star\!,\vx\rangle\!\langle\vx^\star\!,\va_1\!\rangle,
\end{aligned}
\end{eqnarray}
where $\va = [\va_1^\top \ \va_2^\top]^\top\in\RR^{N+K}$ and $\va_1$ lies in the tangent space of the unit sphere at $\vx$, i.e., $\va_1^\top \vx = 0$.

By substituting the properties for the two classes of critical points into~\eqref{Riemannian Hessian total main paper}, we immediately obtain the following result.  
\begin{Theorem}[Characterization of critical points]
\label{thm:riemannian_hessian}
Let $(\vx,\vh)$ be a critical point of problem~\eqref{The loss function of normalization CP PCA real}, and let $\va = [\va_1^\top \ \va_2^\top]^\top \in \mathbb{R}^{N+K}$ with $\va_1^\top \vx = 0$.  
Then, the bilinear form of the (hybrid) Riemannian Hessian satisfies:
\begin{enumerate}
    \item[(1)] For the first class of critical points, $\vx \perp \vx^\star$ and $\vh = \mathbf{0}$,
    \begin{align*}
    \text{Hess} f(\vx,\vh)[\va,\va] = \| \va_2\|_2^2\geq 0;
\end{align*}
\item[(2)] For the second class of critical points, $\vx = \pm\vx^\star$ and $\vh = \vh^\star$,
\begin{align*}
    \text{Hess} f(\vx,\vh)[\va,\va] =  2 \|\vh^\star \|_2^2 \|\va_1\|_2^2  + \|\va_2 \|_2^2\geq 0.
\end{align*}
\end{enumerate}

\end{Theorem}
The first class consists of non-strict saddle points, while the second class consists of global minima.
Consequently, 
any critical point with \(\mathbf h\neq\mathbf0\) satisfies \(\mathbf x=\pm\mathbf x^\star\) and \(\mathbf h=\mathbf h^\star\), yielding exactly the ground-truth tensor \(\tT^\star\).


\subsection{Local Geometry}
\label{sub:local_fac}

The global landscape analysis alone does not guarantee efficient convergence, since first-order algorithms may stagnate near saddle points. We then analyze the local geometry around the global optimum and establish a linear-rate guarantee for RGD in a neighborhood of the ground-truth factors. We apply the following hybrid RGD scheme to problem~\eqref{The loss function of normalization CP PCA real} on the unit sphere:
\begin{equation}
\begin{aligned}
    \label{GD_1 orth CP PCA real main paper}
    \vx_{t+1} &= \text{Retr}_{\vx}(\vx_{t} - \frac{\mu}{2\|\tT^\star\|_F^2} \text{grad}_{\vx}f(\vx_t,\vh_t)),\\
    \vh_{t+1} &= \vh_{t} - \mu \nabla_{\vh}f(\vx_t,\vh_t),
\end{aligned}
\end{equation}
where $\text{Retr}_{\vx}(\widehat{\vx}) = \frac{\widehat{\vx}}{\|\widehat{\vx}\|_2}$ is the standard normalization (polar) retraction on the unit sphere. Here, the subscript $t$ denotes the iteration index, and $\mu>0$ is the step size (learning rate) controlling the update magnitude.


To quantify progress, we measure the deviation of the estimates $\{\vx, \vh\}$ from the ground truth $\{\vx^\star, \vh^\star\}$. The constraint $\|\vx\|_2=1$ removes scalar ambiguity up to a sign  $a \in \{-1,1\}$ because $(a\vx) \circ (a\vx) \circ \vh=\vx \circ \vx \circ \vh$. Define
$\text{dist}^2(\vx,\vh) = \min_{a\in\pm 1}2\|\tT^\star\|_F^2\|\vx - a\vx^\star \|_2^2 
    + \|\vh - \vh^\star\|_2^2,$
where the factor $\|\tT^\star\|_F^2$ balances the two terms since $\|\vx^\star\|_2^2 = 1$ and $\|\vh^\star\|_2^2=\|\tT^\star\|_F^2$.


We first relate this factor metric to the tensor reconstruction error. 
\begin{Lemma}
\label{left ortho upper bound for CP factors real main paper}
Let $\tT = \vx \circ \vx \circ \vh$ and $\tT^\star = \vx^\star \circ {\vx^\star} \circ \vh^\star$ with $\vx,~ \vx^\star\in\RR^N$, $\|\vx\|_2 = \|\vx^\star\|_2 = 1$, and $\vh,~ \vh^\star\in\RR^K$. Assume $\|\vh\|_2\leq \frac{3\|\vh^\star\|_2}{2} = \frac{3\|\tT^\star\|_F}{2}$. Then,
\begin{align*}
   \frac{2}{3}\|\tT - \tT^\star\|_F^2 \leq \text{dist}^2(\vx,\vh) \leq6\|\tT - \tT^\star\|_F^2.
\end{align*}
\end{Lemma}
The detailed proof is provided in the Appendix~\ref{Technical tools used in proofs}. The inequalities in Lemma~\ref{left ortho upper bound for CP factors real main paper} show that, within the neighborhood defined by the assumption $\|\vh\|_2\leq \frac{3\|\vh^\star\|_2}{2} = \frac{3\|\tT^\star\|_F}{2}$, the factor distance and the tensor error are equivalent up to constants. Consequently, linear convergence in one implies linear convergence in the other (with adjusted constants).


We now state the local linear convergence of~\eqref{GD_1 orth CP PCA real main paper}.

\begin{Theorem}[Local linear convergence of RGD]
\label{Local Convergence of Riemannian in the CP PCA_Theorem real}
Let $\tT^\star = \vx^\star \circ {\vx^\star} \circ \vh^\star\in\RR^{N\times N \times K}$ with $\|\vx^\star\|_2=1$ and $\|\vh^\star\|_2>0$.
Suppose the initialization $\{\vx_{0},\vh_{0}  \}$ satisfies $\|\vx_0\|_2=1$ and
$
\operatorname{dist}^2(\vx_0,\vh_0)
\leq
\frac{\|\tT^\star\|_F^2}{288},
$
and choose a stepsize $0<\mu\leq \frac{1}{12}$.
Then, the RGD iterates $\{\vx_{t},\vh_{t}  \}$ obey
 \begin{eqnarray*}
    \label{distance of factors in orth CP PCA derivation final conclusion real}
    \text{dist}^2(\vx_{t+1},\vh_{t+1})\leq(1-\frac{\mu}{32}) \text{dist}^2(\vx_{t},\vh_{t}), \quad t\geq 0.
\end{eqnarray*}
\end{Theorem}

Theorem~\ref{Local Convergence of Riemannian in the CP PCA_Theorem real} shows that, under a mild initialization, RGD converges linearly to the ground-truth factors in a neighborhood of the optimum. The proof is in the Appendix~\ref{Proof of local convergence in the tensor factorization}. 
By Lemma~\ref{left ortho upper bound for CP factors real main paper}, a sufficient condition to guarantee the above initialization condition is $\|\tT_0 - \tT^\star\|_F^2\leq \frac{\|\tT^\star\|_F^2}{1728}$.

\section{Tensor Sensing}
\label{sec:Local Geometry of Noiseless Gaussian Tensor Sensing}

Building on the tensor factorization analysis in
Section~\ref{sec: tensor factorization}, we now turn to the tensor sensing problem, which serves as a more general and analytically tractable surrogate for the structured BDPR model.
The goal is to establish linear convergence of RGD around the ground-truth solution under mild conditions.
Specifically, we consider the recovery of a low-rank tensor $\tT^\star$ from linear measurements as defined in~\eqref{Definition of tensor sensing}.

To guarantee recovery from linear measurements, one typically requires a Restricted
Isometry Property (RIP), widely studied in the compressive sensing literature~\cite{donoho2006compressed, candes2006robust, candes2008introduction}. 
In our setting, since the target tensor $\tT^\star$ admits a CP decomposition $\vx^\star \circ \vx^\star \circ \vh^\star$, which is a special case of the Tucker decomposition with multilinear rank $(1,1,1)$, 
we adapt this notion to the CP model and introduce the following tensor RIP definition~\cite{grotheer2022iterative}, and then invoke a standard result for subgaussian ensembles~\cite[Theorem 2]{rauhut2017low}.
\begin{definition} 
\label{def:cp_rip}
A sensing operator $\mathcal{A}: \mathbb{R}^{N\times N\times K}\rightarrow \mathbb{R}^m$ is said to satisfy the tensor restricted isometry property (TRIP) with constant $\delta_r \in (0,1)$ if 
\begin{eqnarray*}
    (1-\delta_{r})\|\tT\|_F^2\leq \frac{1}{m}\|\mathcal{A}(\tT)\|_2^2\leq(1+\delta_{r})\|\tT\|_F^2
\end{eqnarray*}
holds for all tensors $\tT\in\RR^{N\times N\times K}$ with CP rank at most $r$.
\end{definition}

\begin{Theorem}
\label{RIP condition fro the CP sensing Lemma}
Let $\delta_{r}\in(0,1)$. 
Suppose the sensing tensors $\{\tA_i\}_{i=1}^m$ have i.i.d. subgaussian entries with mean zero and variance one (e.g., Gaussian or Bernoulli).
Then there exists a constant $C>0$ depending only on the subgaussian parameter such that, for any $\epsilon \in (0,1)$, if
\begin{equation}
m \ge C \cdot \frac{1}{\delta_{r}^2} \cdot \max\left\{r^3 + (N + K)r, \log(1/\epsilon)\right \},
\label{eq:mrip ToT CP}
\end{equation}
the linear operator $\calA$ obeys the TRIP with constant $\delta_r$ 
for every CP tensor $\tT\in\RR^{N\times N\times K}$ of rank at most $r$ $(r\leq \min\{N,K\})$, 
with probability at least $1-\epsilon$. This includes, as a special case, the rank-one structure $\tT = \vx \circ \vx \circ \vh$.
\end{Theorem}
Theorem~\ref{RIP condition fro the CP sensing Lemma} implies that, with $r=1$ in our setting, a number of measurements $m$ scaling linearly with $N+K$ suffices for the normalized measurement energy $ \frac{1}{m}\|\mathcal{A}(\tT)\|_2^2$ to remain proportional to $\|\tT\|_F^2$. Such RIP-type guarantees are well established for subgaussian ensembles~\cite{qin2024guaranteed}. 
We leave the formal proof of TRIP for the structured sensing operator~\eqref{def_senT} used in the original BDPR problem as future work.


Given the measurements $\vy = \calA(\tT^\star)$, we recap the factorized rank-one partially symmetric tensor recovery problem as follows
\begin{equation}
\label{loss_sensing}
\begin{aligned}
\min_{\substack{\vx \in \RR^N\!\!, \|\vx\|_2=1 \\ \vh \in \RR^K}} 
\!\!\!\!\!\!g(\vx,\vh) \! &= \! \frac{1}{2m} \!
\|\calA(\vx \circ \vx \circ \vh) \!-\! \vy\|_2^2. \\
\end{aligned}
\end{equation}

\subsection{RGD Converges to a Global Solution at a Linear Rate}
Following the analysis for the factorization problem in the last section, we compute the Riemannian gradient of $g(\vx,\vh)$ with respect to $\vx$ as
$\text{grad}_{\vx}g(\vx,\!\vh) =\calP_{\text{T}_{\vx} \text{St}}(\nabla_{\vx}g(\vx,\vh)),$
where
\begin{align*}
\nabla\!_{\vx}g(\vx,\!\vh)\!
= &\frac{1}{m}\!\sum_{i=1}^{m} \!\big(\!\big\langle \tA_i, \vx \!\circ\! \vx \!\circ\! \vh\big\rangle \!-\! \vy(i) \big) \!\times \!\big( 
\calM_1(\tA_i)(\vh\otimes\vx) 
\!+ \!\calM_2(\tA_i)(\vx\otimes\vh) 
\big)
\end{align*}
denotes the Euclidean gradient. In addition, the Euclidean gradient with respect to $\vh$ is 
$$
\nabla _{\vh}g(\vx, \vh)
=\frac{1}{m}\sum_{i=1}^{m} \big( \big\langle  \tA_i,  \vx \circ \vx \circ \vh \big\rangle  -  \vy(i)  \big)   \calM_3(\tA_i)(\vx\otimes \vx). 
$$


As in \eqref{GD_1 orth CP PCA real main paper}, we employ the hybrid RGD updates:
\begin{equation}
\begin{aligned}
    \label{GD_1 orth CP sensing real main paper}
    \vx_{t+1} &= \text{Retr}_{\vx}(\vx_{t} - \frac{\mu}{2\|\tT^\star\|_F^2} \text{grad}_{\vx}g(\vx_t,\vh_t)),\\
    \vh_{t+1} &= \vh_{t} - \mu \nabla_{\vh}g(\vx_t,\vh_t).
\end{aligned}
\end{equation}

\subsubsection{Spectral initialization}
Let $\calA^*$ denote the adjoint of $\calA$, i.e., $\calA^*(\vy) = \sum_{i=1}^m y_i \, \tA_i .$ We initialize the RGD algorithm by
\begin{equation}
    \begin{aligned}
    \label{SVD of X spectral initialization}
    \vx_0 = \vu,\quad
    \vh_0 = \calM_3(\frac{1}{m}\calA^*(\vy))(\vx_0\otimes \vx_0).
    \end{aligned}
\end{equation}
Here, $\vu$ is a leading left singular vector of
$\calM_1(\frac{1}{m}\calA^*(\vy))$.
The accuracy of this spectral initializer is quantified in Theorem~\ref{TENSOR SENSING SPECTRAL INITIALIZATION CP} below.

\begin{Theorem}[Spectral initializer accuracy]
\label{TENSOR SENSING SPECTRAL INITIALIZATION CP}
If the linear operator $\calA$ satisfies the TRIP for CP tensors with $r=2$, then the spectral initialization in~\eqref{SVD of X spectral initialization} obeys

\begin{equation}
\label{TENSOR SENSING SPECTRAL INITIALIZATION CP1}
\|\tT_0-\tT^\star\|_F
\leq
3\sqrt{K}\,\delta_r
\|\tT^\star\|_F.
\end{equation}
\end{Theorem}
The proof is given in the Appendix~\ref{Proof of spectral initialization in the noiseless}. Thus, for a sufficiently small TRIP level $\delta_r$, the initialization lies within a controlled neighborhood of the ground truth.

\subsubsection{Local linear convergence}
The following theorem provides a local linear convergence result for RGD.
\begin{Theorem}[Local linear convergence of RGD]
\label{Local Convergence of Riemannian in the CP sensing_Theorem real}
Let $\tT^\star = \vx^\star \circ {\vx^\star} \circ \vh^\star\in\RR^{N\times N \times K}$ {with $\|\vx^\star\|_2=1$ and $\|\vh^\star\|_2>0$}.
Suppose $\calA$ satisfies the TRIP with
{$r=2$ and $0\leq\delta_r<\frac{9}{20}$}.
If the initialization $\{\vx_0,\vh_0\}$ satisfies
{$\|\vx_0\|_2=1$ and}
{
\[
\operatorname{dist}^2(\vx_0,\vh_0)\leq\frac{\|\tT^\star\|_F^2}{400},
\]
}
and the step size obeys
{$0<\mu\leq(9-20\delta_r)/(40(1+\delta_r)^2)$},
then the RGD iterates $\{\vx_t,\vh_t\}$
{generated by~\eqref{GD_1 orth CP sensing real main paper}}
satisfy
{
\[
\operatorname{dist}^2(\vx_{t+1},\vh_{t+1})
\leq\left(1-\frac{9-20\delta_r}{10}\mu\right)\operatorname{dist}^2(\vx_t,\vh_t),\qquad t\geq0.
\]
}
\end{Theorem}
The proof is provided in the Appendix~\ref{Proof of local convergence of CP sensing}. 
Under TRIP, the favorable local geometry ensures linear convergence of RGD, with the convergence rate degrading smoothly as $\delta_r$ increases. By Lemma~\ref{left ortho upper bound for CP factors real main paper}, a sufficient initialization condition is
$\|\tT_0-\tT^\star\|_F^2\leq\|\tT^\star\|_F^2/2400$.\footnote{Indeed, $\|\vh_0\|_2=\|\tT_0\|_F\leq\|\tT^\star\|_F+\|\tT_0-\tT^\star\|_F<3\|\tT^\star\|_F/2$, so the lemma applies and yields the required factor-distance bound.}
By Theorem~\ref{TENSOR SENSING SPECTRAL INITIALIZATION CP}, the spectral initialization in~\eqref{SVD of X spectral initialization} satisfies this condition whenever $\delta_r\leq1/\sqrt{21600K}$ with $r=2$. Consequently, for $N,K\geq2$ and any $\epsilon\in(0,1)$, combining this with~\eqref{eq:mrip ToT CP} yields the convergence guarantee in Theorem~\ref{Local Convergence of Riemannian in the CP sensing_Theorem real} for the subgaussian sensing ensemble, with probability at least $1-\epsilon$, provided
\(
m\geq C K\max\{N+K,\log(1/\epsilon)\},
\)
where $C>0$ depends only on the subgaussian parameter.

\subsection{Extension to Noisy Case}

\label{sec: noisy tensor sensing}


In practical tensor recovery problems, measurements may be corrupted by noise~\cite{qin2024quantum,qin2024guaranteed,cai2019nonconvex}. In this section, we consider recovering $\tT^\star$ from noisy linear measurements of the form
\begin{eqnarray}
    \vy = \calA(\tT^\star)+ \ve,
    \label{eqn:noisy_y}
\end{eqnarray}
where $\ve\sim\mathcal N(\mathbf 0,\gamma^2 I_m)$ is independent of $\calA$.
We estimate $\tT^\star$ by solving the constrained least-squares problem~\eqref{loss_sensing} with noisy measurements~\eqref{eqn:noisy_y}. We employ the same hybrid RGD scheme as in the noiseless case (see \eqref{GD_1 orth CP sensing real main paper}) with the spectral initialization in~\eqref{SVD of X spectral initialization}.
The following result extends the guarantee in Theorem~\ref{TENSOR SENSING SPECTRAL INITIALIZATION CP} to the noisy case.
\begin{Theorem}[Spectral initializer accuracy under noise]
\label{TENSOR SENSING SPECTRAL INITIALIZATION CP noisy}
If $\calA$ satisfies the TRIP for CP tensors with $r=2$,
then the initializer in~\eqref{SVD of X spectral initialization} obeys
\begin{equation*}
\begin{aligned}
\|\tT_0-\tT^\star\|_F
\leq3\sqrt K\,\delta_r\|\tT^\star\|_F+
O\left(\gamma\sqrt{\frac{K(1+\delta_r)(N+K)}{m}}\right)
\end{aligned}
\end{equation*}
with probability at least $1-2\exp(-(2N+K))$, conditional on the TRIP.
\end{Theorem}

This extends the noiseless guarantee in~\eqref{TENSOR SENSING SPECTRAL INITIALIZATION CP1} by an additive term induced by measurement noise. The proof is in the Appendix~\ref{Proof of spectral initialization in the noise}.


Similarly, we obtain the following noise-robust local convergence guarantee that extends 
Theorem~\ref{Local Convergence of Riemannian in the CP sensing_Theorem real}.
\begin{Theorem}[Local linear convergence of RGD under noise]
\label{Local Convergence of Riemannian in the CP sensing_Theorem real noisy}
Let $\tT^\star=\vx^\star\circ\vx^\star\circ\vh^\star\in\RR^{N\times N\times K}$ with $\|\vx^\star\|_2=1$ and $\|\vh^\star\|_2>0$. Suppose that the linear operator $\calA$ satisfies the TRIP with $r=2$ and $0\leq\delta_r<2/5$. If the initialization $\{\vx_0,\vh_0\}$ satisfies $\|\vx_0\|_2=1$ and
\begin{equation}
\label{intial requirement in the sensing real noisy}
\begin{aligned}
\operatorname{dist}^2(\vx_0,\vh_0)&\leq \frac{\|\tT^\star\|_F^2}{400},
\end{aligned}
\end{equation}
and the step size obeys $0<\mu\leq(2-5\delta_r)/(20(1+\delta_r)^2)$, then the iterates $\{\vx_t,\vh_t\}$ generated by RGD satisfy the following bound simultaneously for all $t\geq0$, with probability at least $1-2\exp(-(2N+K))$ over the noise, conditional on $\calA$ satisfying the stated TRIP:
\[
\begin{aligned}
\operatorname{dist}^2(\vx_{t+1},\vh_{t+1}) \leq\left(1-\frac{2(2-5\delta_r)}5\mu\right)^{t+1}\operatorname{dist}^2(\vx_0,\vh_0) +O\left(\frac{(N+K)\gamma^2}{m(2-5\delta_r)}\right),
\end{aligned}
\]
provided $m\geq C(N+K)\gamma^2/[(2-5\delta_r)\|\tT^\star\|_F^2]$ for a sufficiently large universal constant $C>0$. Here, $\gamma^2$ denotes the noise variance.
\end{Theorem}

Theorem~\ref{Local Convergence of Riemannian in the CP sensing_Theorem real noisy} demonstrates graceful degradation in the presence of noise: RGD maintains a linear convergence rate up to a noise floor that scales proportionally with the noise level $\gamma$, with constants depending smoothly on $\delta_r$. The proof is provided in the Appendix~\ref{Proof of local convergence of CP sensing noisy}. This result closely parallels the noiseless case in Theorem~\ref{Local Convergence of Riemannian in the CP sensing_Theorem real}, demonstrating that---once the initialization condition is satisfied---RGD converges linearly to a neighborhood of the ground truth, whose radius matches the statistical error induced by noise. Furthermore, using Lemma~\ref{left ortho upper bound for CP factors real main paper}, a sufficient condition to guarantee the initialization condition~\eqref{intial requirement in the sensing real noisy} is
$\|\tT_0-\tT^\star\|_F^2\leq\frac{\|\tT^\star\|_F^2}{2400}.$

\section{Numerical Experiments}
\label{sec:Experiments}

In this section, we conduct a series of experiments to further validate the theoretical results for the tensor sensing problem~\eqref{loss_sensing} and its noisy version, using both Gaussian and structured sensing tensors defined in~\eqref{def_senT}, the latter corresponding to recovery in the original BDPR problem.

\subsection{Recovery with  Noiseless Gaussian Measurements}\label{sec: exp_GauTen}

In this experiment, we apply the RGD method~\eqref{GD_1 orth CP sensing real main paper} with spectral initialization~\eqref{SVD of X spectral initialization} to solve the tensor sensing problem~\eqref{loss_sensing}, where the sensing tensors are generated as random Gaussian tensors with entries following $\mathcal{N}(0,1)$. The ground truth tensor is given by $\tT^\star = \vx^\star \circ \vx^\star \circ \vh^\star \in \RR^{N \times N \times K}$ with $N=10$ and $K=6$, where $\vx^\star\in \RR^N$ is a normalized standard Gaussian vector and $\vh^\star\in \RR^K$ is a standard Gaussian vector (i.e., entries sampled independently from $\mathcal{N}(0,1)$). The number of measurements is set to $m=60$, and the algorithm is executed with a fixed step size of 0.05 for 2500 iterations.
As can be seen in Figure \ref{fig:real_case}, all three plots confirm linear convergence toward the ground-truth solution, consistent with our theoretical results.

\begin{figure}[t]
    \centering
    \begin{tabular}{ccc}
    \subfloat[]{\includegraphics[width=0.3\linewidth]{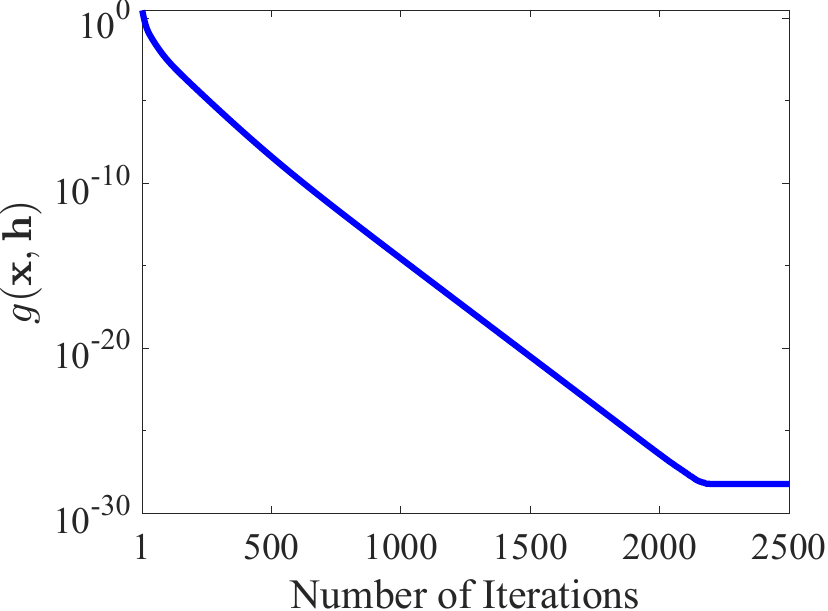}} & 
    \subfloat[]{\includegraphics[width=0.3\linewidth]{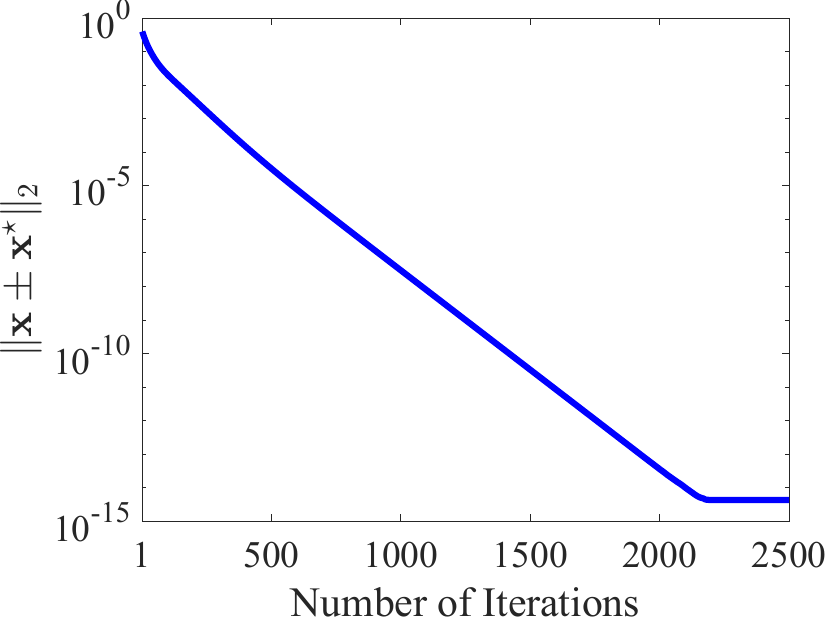}} & 
    \subfloat[]{\includegraphics[width=0.3\linewidth]{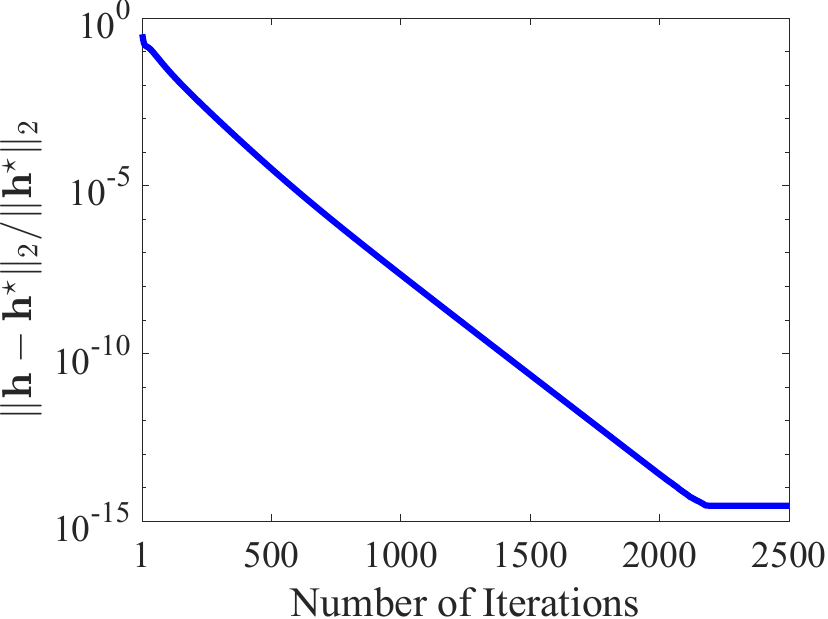}}
    \end{tabular}
    \caption{Convergence of RGD for noiseless Gaussian tensor sensing. (a) Loss $g(\vx,\vh)$.  (b) Reconstruction error of signal $\vx$ measured by $\|\vx \pm \vx^\star\|_2$. (c) Relative error of the coefficient vector $\vh$: $\frac{\|\vh - \vh^\star\|_2}{\|\vh^\star\|_2}$.}
    \label{fig:real_case}
\end{figure}


We further evaluate recovery rates as $m$ varies, using the same initialization as described above, and declaring success when
$\min\{\|\vx-\vx^\star\|_2,\|\vx+\vx^\star\|_2\}\le10^{-5}$ and
$\frac{\|\vh - \vh^\star\|_2}{\|\vh^\star\|_2} \leq 10^{-5}$.

\begin{figure}[!htbp]
    \centering
    \begin{minipage}{0.84\textwidth}
        \centering
        \subfloat[$K=6$]{%
            \includegraphics[width=0.35\linewidth]
            {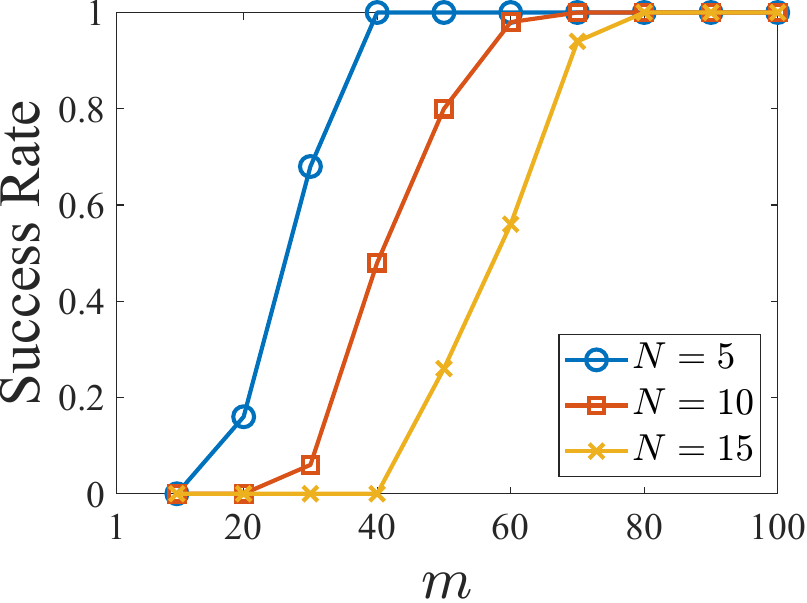}%
        }%
        \hspace{0.02\linewidth}%
        \subfloat[$N=10$]{%
            \includegraphics[width=0.35\linewidth]
            {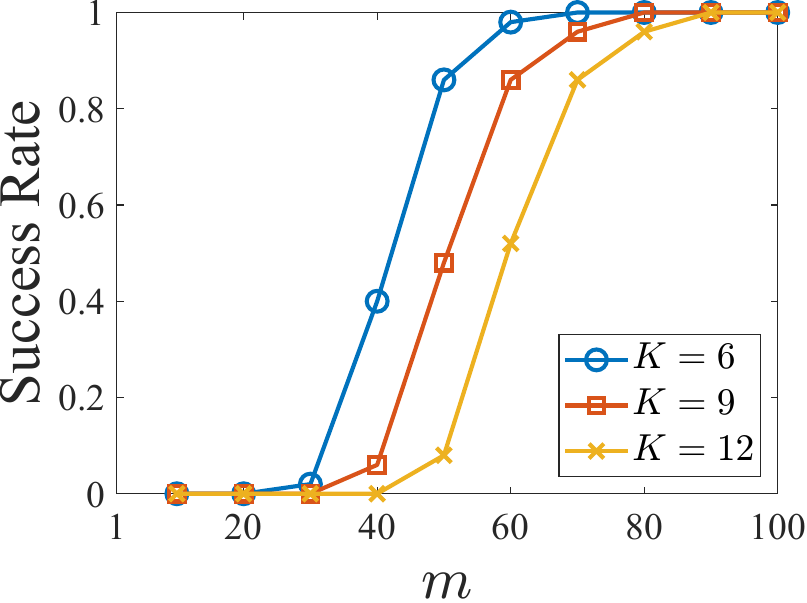}%
        }
    \end{minipage}
    \caption{Successful recovery rates of RGD versus the number of measurements $m$.}
    \label{fig:success_rate}
\end{figure}

We vary $N\in\{5,10,15\}$ with $K=6$ and $K\in\{6,9,12\}$ with $N=10$, using 100 iterations, step size $0.5$, and the success rate is computed over 50 independent trials.
Figure~\ref{fig:success_rate} shows that the success rate improves as the number of measurements increases. In addition, lower values of the signal dimension $N$ or subspace dimension $K$ consistently yield higher successful recovery rates, highlighting the influence of problem complexity on sample efficiency.

\subsection{Recovery with Noiseless Structured Measurements}

We repeat the experiments using structured sensing tensors~\eqref{def_senT} with $N=25$, $K=6$, and $I=60$. The ground truth $\vx^\star \in \CC^{N}$ is a normalized standard complex Gaussian vector, while $\vh^\star \in \RR^{K}$ is generated as in Section~\ref{sec: exp_GauTen}. We draw $\mB \in \RR^{N \times K}$ and $\mP_i \in \RR^{N \times N}$ as standard Gaussian matrices with entries drawn from $\mathcal{N}(0,1)$. Following~\cite{li2019simultaneous}, we replace the chirp masks $\vg_i$ by complex Gaussian vectors with i.i.d. $\mathcal{CN}(0,1)$ entries; $\tA_{i,n}$ and $\vy$ are generated according to~\eqref{def_senT} and~\eqref{eqn:measurement}, respectively.

We then apply RGD with spectral initialization~\eqref{SVD of X spectral initialization} to solve the tensor sensing problem~\eqref{loss_sensing}, running 25000 iterations with fixed step sizes of 40 for $\vx$ and 8 for $\vh$. We also replace $\|\tT^\star\|_F^2$ with $\|\tT_t\|_F^2$ in the RGD updates~\eqref{GD_1 orth CP sensing real main paper}.
Figure~\ref{fig:phase_imaging} illustrates the convergence behavior of the algorithm. While our current theory does not provide a formal proof of local linear convergence, the numerical results consistently exhibit such behavior. This motivates our future work aimed at analyzing the local landscape of the tensor sensing problem~\eqref{loss_sensing} with the structured sensing tensors defined in~\eqref{def_senT}.

\begin{figure}[tbp]
    \centering
    \begin{tabular}{ccc}
    \subfloat[]{\includegraphics[width=0.3\linewidth]{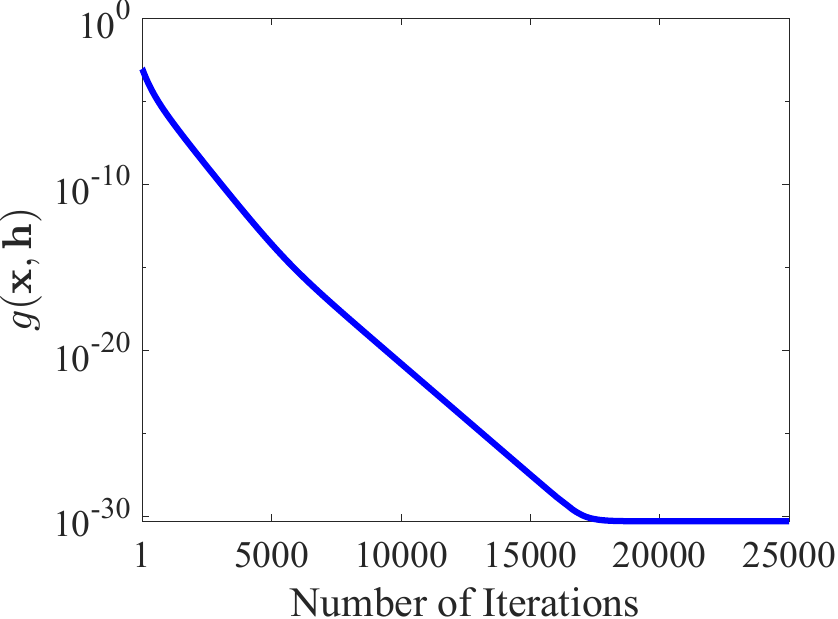}} & 
    \subfloat[]{\includegraphics[width=0.3\linewidth]{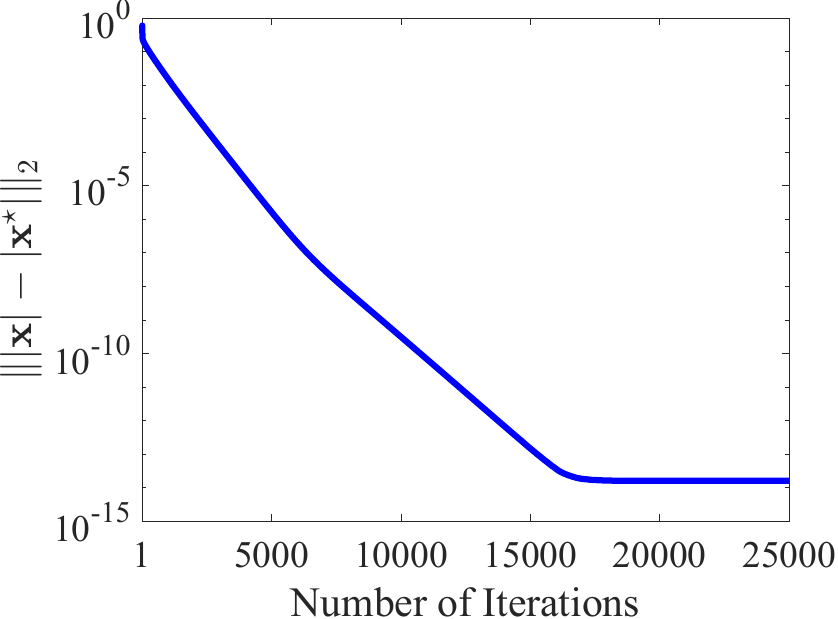}} & 
    \subfloat[]{\includegraphics[width=0.3\linewidth]{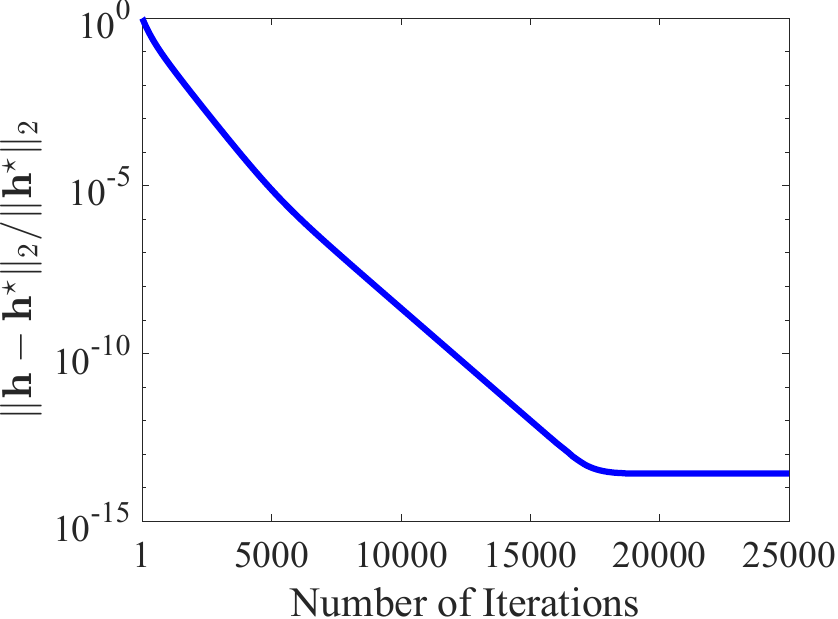}}
    \end{tabular}
    \caption{Convergence of RGD for noiseless structured tensor sensing. (a) Loss $g(\vx,\vh)$.  (b) Reconstruction error of signal $\vx$ measured by $\||\vx|-|\vx^\star|\|_2$. (c) Relative error of the coefficient vector $\vh$: $\frac{\|\vh - \vh^\star\|_2}{\|\vh^\star\|_2}$. }
    \label{fig:phase_imaging}
\end{figure}

\begin{figure}[!htbp]
    \centering
    \begin{minipage}{0.88\textwidth}
        \centering
        \subfloat[$K=6$]{%
            \includegraphics[width=0.35\linewidth]
            {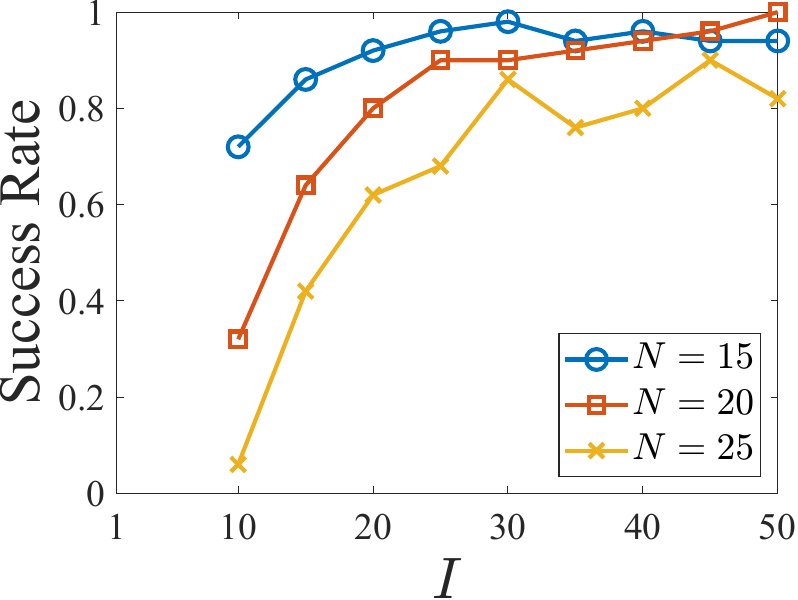}%
        }%
        \hspace{0.02\linewidth}%
        \subfloat[$K=6$]{%
            \includegraphics[width=0.35\linewidth]
            {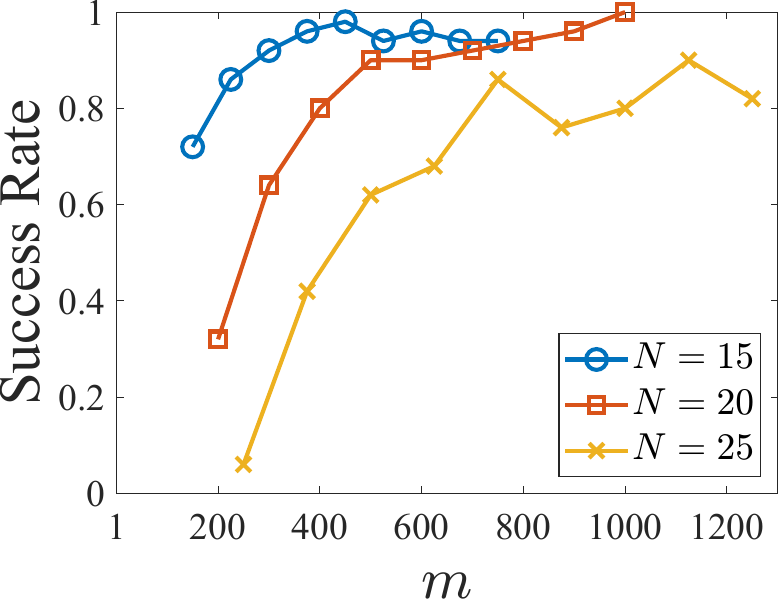}%
        }%
        \par\vspace{4pt}
        \subfloat[$N=20$]{%
            \includegraphics[width=0.35\linewidth]
            {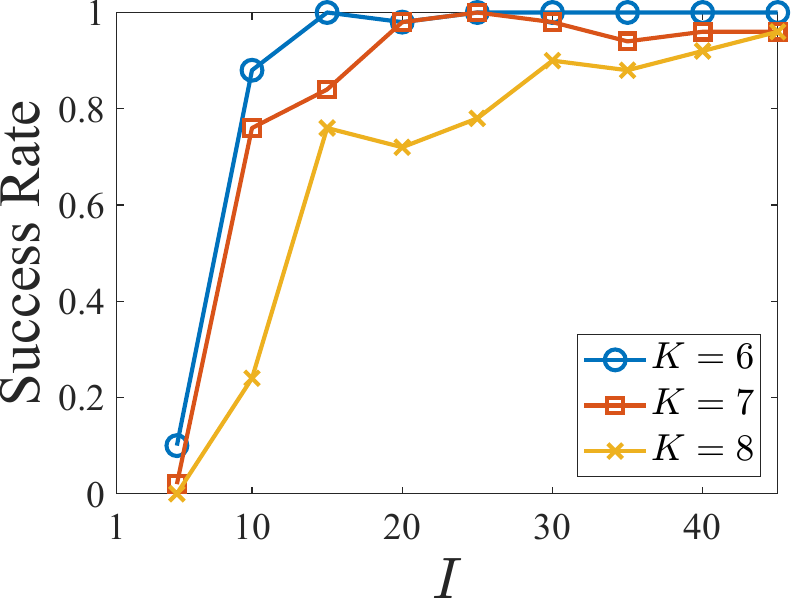}%
        }%
        \hspace{0.02\linewidth}%
        \subfloat[$N=20$]{%
            \includegraphics[width=0.35\linewidth]
            {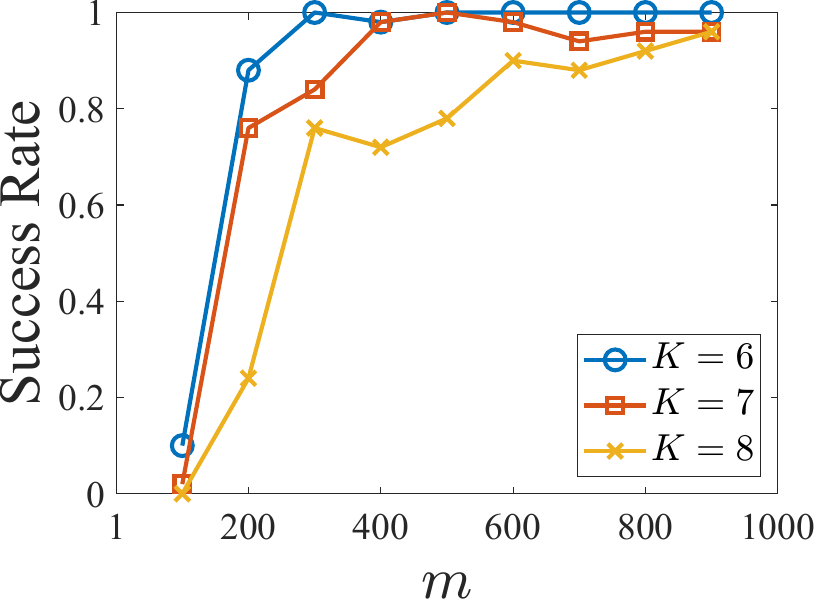}%
        }
    \end{minipage}
    \captionsetup{font=small,skip=3pt}
    \caption{Successful recovery rates versus detector positions $I$
    and measurements $m$.}
    \label{fig:structured_success_rate}
\end{figure}

Figure~\ref{fig:structured_success_rate} presents the successful recovery rates under varying numbers of detector positions $I$  and number of measurements $m$ ($ m = N I$).

We use step sizes $40$ for $\vx$ and $8$ for $\vh$, with a maximum number of 11000 iterations for fixed $K$ and 20000 iterations for fixed $N$. 
Recovery is successful if:
\footnote{For a complex-valued signal, the natural ambiguity is a global phase rather than a sign. Here we report the magnitude error \(\||\mathbf{x}|-|\mathbf{x}^\star|\|_2\) for the structured BDPR experiments.}
$\||\vx|-|\vx^\star|\|_2\leq10^{-5}$ and
$\frac{\|\vh - \vh^\star\|_2}{\|\vh^\star\|_2} \leq 10^{-5}$.
The success rate increases with detector positions $I$ or measurements $m$. These results exhibit a similar trend as in the random Gaussian setting (Figure~\ref{fig:success_rate}), while also demonstrating that recovery remains robust under more structured and realistic sensing models. In particular, smaller values of  $N$ and $K$ continue to facilitate successful recovery, underscoring the importance of underlying problem complexity in structured tensor sensing.


\subsection{Recovery with Noisy Measurements}

In real-world sensing systems, measurements are inevitably contaminated by noise due to sensor imperfections, environmental disturbances, or hardware limitations. Thus, we revisit the tensor sensing problem using the noisy observation model~\eqref{eqn:noisy_y}, incorporating Gaussian noise with standard deviations $\gamma = \{0, 0.001, 0.01, 0.1\}$. We evaluate both Gaussian sensing tensors and the structured sensing tensors defined in~\eqref{def_senT}, with fixed dimensions $N = 20$ and $K = 6$. For the Gaussian tensor sensing case, we apply the RGD method~\eqref{GD_1 orth CP sensing real main paper} using a fixed step size $\mu = 0.5$ for 50 iterations. As shown in Figure~\ref{fig:gaussian_noise}, the objective function $g(\vx,\vh)$ and reconstruction errors of the signal component $\vx^\star$ and subspace coefficient $\vh^\star$ consistently decrease as the number of measurements increases. Moreover, the performance degrades gracefully with higher noise levels. 
In the structured tensor setting, we employ the RGD updates in~\eqref{GD_1 orth CP sensing real main paper} with different step sizes for $\vx$ and $\vh$. Specifically, the update of $\vx$ uses a base step size $\mu = 15$, while the update of $\vh$ uses $\mu = 6$. To account for varying noise levels $\gamma$, we scale both step sizes by a factor $s(\gamma)$, chosen from $\{1, 0.8, 0.5, 0.2\}$ for $\gamma \in \{0, 0.001, 0.01, 0.1\}$, respectively. Here, we also replace $\|\tT^\star\|_F^2$ with $\|\tT_t\|_F^2$ in the RGD updates~\eqref{GD_1 orth CP sensing real main paper}. The algorithm is run for at most 3000 iterations.
Figure~\ref{fig:phaseimaging_noise} presents the loss and estimation errors under various noise levels and number of detector positions. Similar trends are observed: increasing the number of detector positions can significantly reduce the estimation error. For any fixed number of detector positions, lower noise levels consistently yield smaller estimation errors.

\begin{figure}[!t]
    \centering
    \begin{tabular}{ccc}
    \subfloat[]{\includegraphics[width=0.3\linewidth]{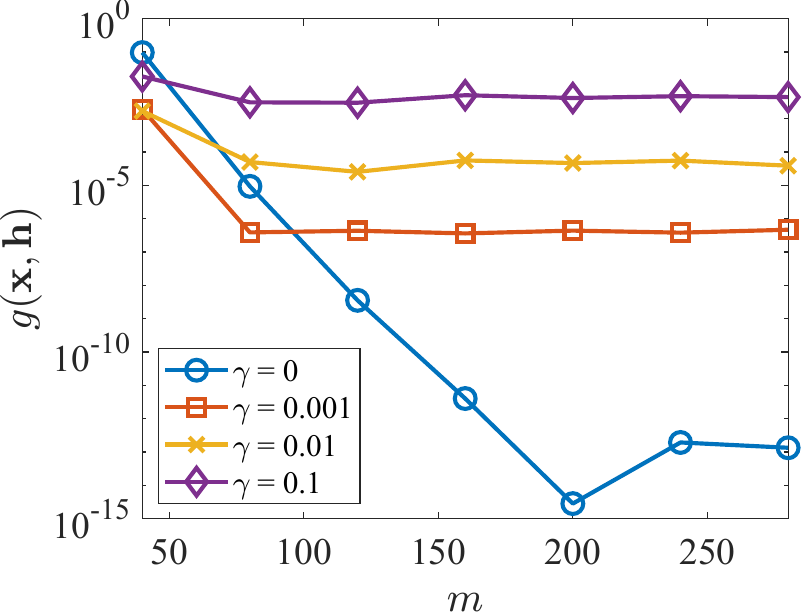}} & 
    \subfloat[]{\includegraphics[width=0.3\linewidth]{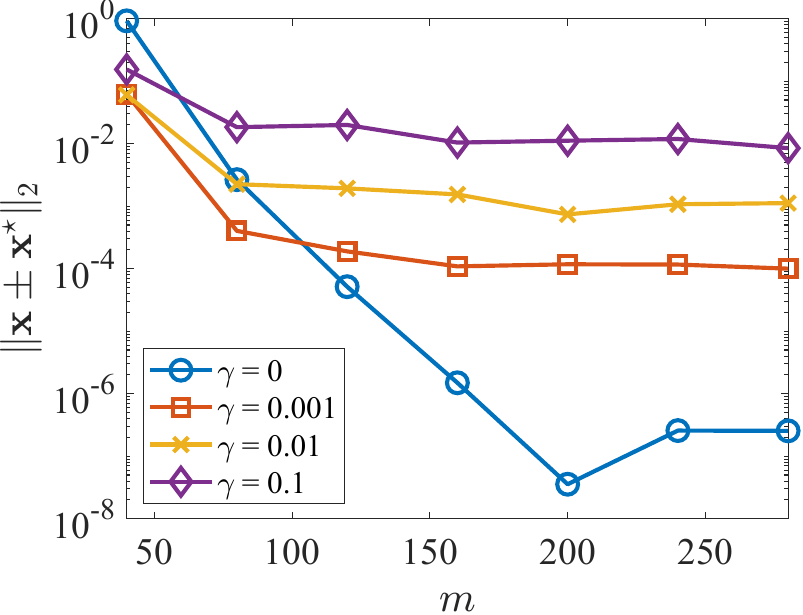}} & 
    \subfloat[]{\includegraphics[width=0.3\linewidth]{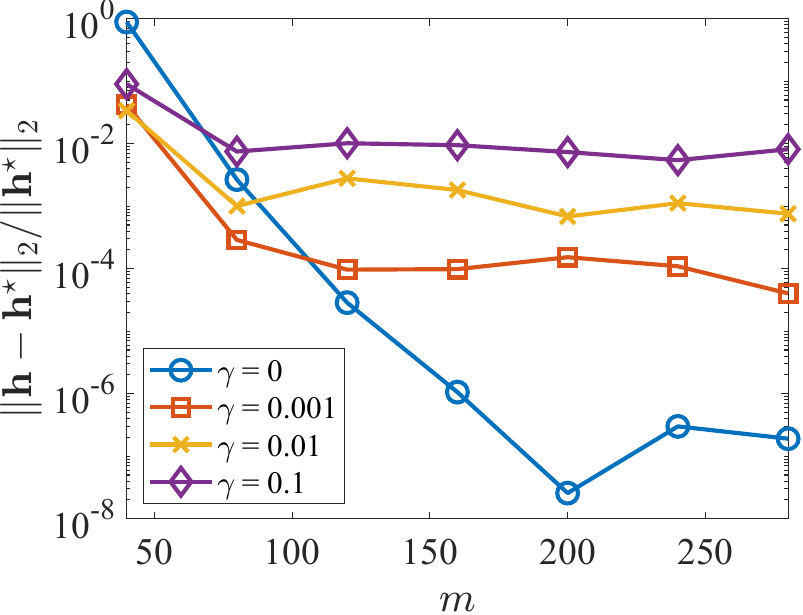}}
    \end{tabular}
    \caption{Gaussian tensor sensing with noise: the loss function $g(\vx,\vh)$ and reconstruction errors of $\vx^\star $ and $\vh^\star $ under various noise levels and number of measurements. }
    \label{fig:gaussian_noise}
\end{figure}

\begin{figure}[!t]
    \centering
    \begin{tabular}{ccc}
    \subfloat[]{\includegraphics[width=0.3\linewidth]{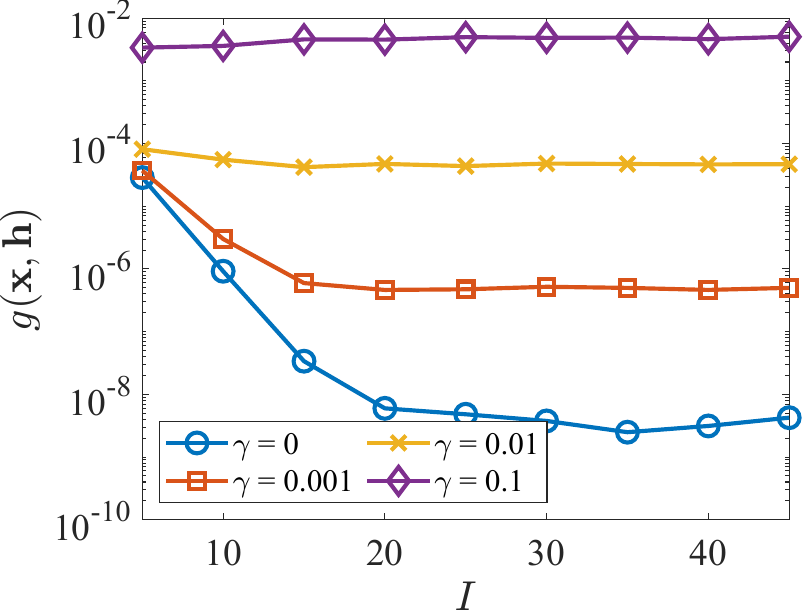}} & 
    \subfloat[]{\includegraphics[width=0.3\linewidth]{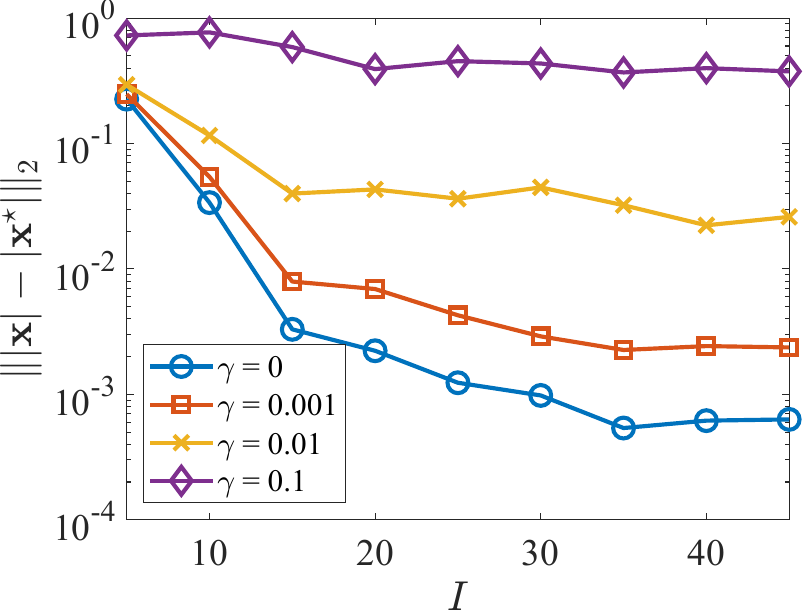}} & 
    \subfloat[]{\includegraphics[width=0.3\linewidth]{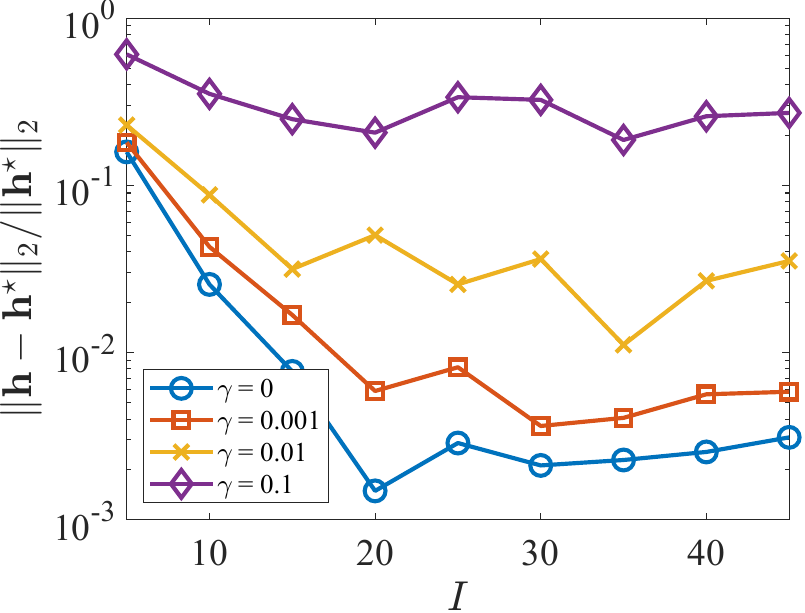}}
    \end{tabular}
    \caption{Structured tensor sensing with noise: the loss function $g(\vx,\vh)$ and reconstruction errors of $\vx^\star $ and $\vh^\star $ under various noise levels and number of detector positions.}
    \label{fig:phaseimaging_noise}
\end{figure}


\section{Conclusion}
\label{sec:Conclusion}

In this work, we studied the landscape of the BDPR problem through the perspective of structured low-rank tensor recovery. While the original BDPR formulation can be recast as a low-rank tensor recovery problem, the intricate structure of the associated sensing tensor makes a direct analysis intractable. To address this challenge, 
we considered tractable surrogates, starting from a tensor factorization problem (the population risk of tensor sensing) and extending to the tensor sensing formulation.
For the tensor factorization setting, we fully characterized the optimization geometry, identifying all critical points and establishing convergence guarantees for Riemannian gradient descent. We further extended these results to the tensor sensing scenario, demonstrating that favorable geometric properties persist under appropriate conditions. In addition, we established robustness guarantees under measurement noise, showing that the fundamental geometric structure remains stable even with corrupted observations. These findings provide valuable insights into the optimization landscape of the original BDPR problem and offer principled guidance for the design of efficient algorithms. We leave a direct characterization of the optimization landscape of the original BDPR problem, together with a formal proof of TRIP for its structured sensing model, as important directions for future research.

\section*{Acknowledgements}
This work was supported in part by NSF grants ECCS-2409701, ECCS-2409702, and CCF-2241298. ZQ gratefully acknowledges support from the MICDE Research Scholars Program at the University of Michigan.

\newpage
\setlength{\bibsep}{4pt}
\bibliographystyle{elsarticle-num}
\bibliography{reference}

\clearpage
\appendix
\renewcommand{\thesection}{\Alph{section}}

\section[Derivation of the Riemannian Hessian at Critical Points]{Derivation of the Riemannian Hessian at the Critical Points of Problem~\eqref{The loss function of normalization CP PCA real}}
\label{Proof of Hessian matrix in the appendix}

Recall that the Euclidean gradients of $f(\vx,\vh)$ with respect to $\vx$ and $\vh$ are
\begin{align*}
    \nabla_{\vx}f(\vx,\vh) &= 2 \|\vh \|_2^2\|\vx\|_2^2\vx - 2\<\vh^\star,\vh  \>\<\vx^\star, \vx \>\vx^\star,\\
    \nabla_{\vh}f(\vx,\vh) &= \|\vx\|_2^4\vh - \< \vx^\star,\vx \>^2\vh^\star.
\end{align*}
At any critical point $(\vx,\vh)$, the Riemannian gradient with respect to $\vh$ vanishes, which is equivalent to $\nabla_{\vh}f(\vx,\vh) = \vzero$. Hence,
\begin{align}
    \<\nabla_{\vh}f(\vx,\vh), \vh  \>  = \|\vx\|_2^4\|\vh\|_2^2 - \< \vx^\star,\vx \>^2\<\vh^\star, \vh  \> = 0,
    \label{eq:nabla_h,h}
\end{align}
which implies $\<\vh^\star, \vh  \> \geq 0$. Moreover, 
\begin{align*}
    \<\nabla_{\vx}f(\vx,\vh), \vx  \>  = 2 \|\vh \|_2^2\|\vx\|_2^4 - 2\<\vh^\star,\vh  \>\<\vx^\star, \vx \>^2,
\end{align*}
which, by substituting~\eqref{eq:nabla_h,h},
reduces to  
\begin{align}
\label{inner product of GD_x and x}
  \<\nabla_{\vx}f(\vx,\vh), \vx  \> =0.  
\end{align}

We now compute the bilinear form of the (hybrid) Riemannian Hessian at a critical point. Throughout, let $\va = [\va_1^\top \ \va_2^\top]^\top \in \RR^{N+K}$ with $\va_1 \in \text{T}_{\vx} \text{St}$ (so $\va_1^\top \vx = 0$) and $\va_2\in \RR^K$. Denote by $\calP_{\text{T}_{\vx} \text{St}}(\cdot ) = (\mId - \vx\vx^\top)(\cdot)$ the orthogonal projection onto the tangent space
of the unit sphere at $\vx$, i.e., $\text{T}_{\vx} \text{St}$.



The Euclidean directional Hessian with respect to $\vx$ is 
\begin{align*}
\label{Hessian xx}
\nabla_{\vx\vx^\top}^2 \!f(\vx,\vh)[\va_1] &= \lim_{t\to 0} \frac{1}{t}\Big(
     \nabla_{\vx}f(\vx + t \va_1,\vh) -\nabla_{\vx}f(\vx,\vh) \Big)  \\
 &= 2\|\vh\|_2^2 \|\vx\|_2^2 \va_1 - 2\langle \vh^\star, \vh \rangle \langle \vx^\star, \va_1 \rangle \vx^\star.
\end{align*}
Using~\eqref{inner product of GD_x and x}, we have
\begin{align*}
    \text{Hess}_{\vx\vx^\top} f(\vx,\vh)[\va_1]
=&\calP_{\text{T}_{\vx} \text{St}}(\nabla_{\vx\vx}^2 f(\vx,\vh)[\va_1]-\va_1\vx^\top \nabla_{\vx}f(\vx, \vh))\\
    =&\calP_{\text{T}_{\vx} \text{St}}(\nabla_{\vx\vx}^2 f(\vx,\vh)[\va_1]).
\end{align*}
Then, the Riemannian Hessian bilinear form becomes
\begin{align*}
   &\text{Hess}_{\vx\vx^\top} f(\vx,\vh)[\va_1,\va_1]
         =    \<\va_1, \text{Hess}_{\vx\vx^\top} f(\vx,\vh)[\va_1] \>
     \\   
       =     &\<\va_1,\nabla_{\vx\vx^\top}^2 f(\vx,\vh)[\va_1]\>
        =     2 \|\vh   \|_2^2 \|\va_1\|_2^2     -     2\<\vh^\star  ,  \vh    \>  \<\vx^\star  ,   \va_1   \>^2,
\end{align*}
where we used $\calP_{\text{T}_{\vx} \text{St}}(\va_1) = \va_1$ and $\|\vx\|_2 = 1$. 

Similarly, we can get the Euclidean Hessian bilinear form with respect to $\vh$: 
\begin{equation*}
\label{Hessian hh}
\begin{aligned}
\nabla_{\vh\vh^\top}^2 f(\vx,\vh)[\va_2, \va_2] = 
\|\va_2\|_2^2.
\end{aligned}
\end{equation*}

Recall that the Riemannian gradient of $f(\vx,\vh)$ with respect to $\vx$ is 
\begin{align*}
\text{grad}_{\vx}f(\vx,\vh) 
    = - 2\<\vh^\star,\vh  \>\<\vx^\star, \vx \>\vx^\star   + 2\<\vh^\star,\vh  \>\<\vx^\star, \vx \>^2\vx.
\end{align*}
Differentiating with respect to $\vh$ in the direction $\va_2$ and pairing with $\va_1$ gives
\begin{equation*}
\label{Hessian hx another}
\begin{aligned}
&\nabla_{\vh^\top} \text{grad}_{\vx} f(\vx,\vh)[\va_1,\va_2] \\
=& \left\langle \va_1,\; \lim_{t\to 0} \frac{ \text{grad}_{\vx} f(\vx,\vh+ t \va_2) - \text{grad}_{\vx} f(\vx,\vh) }{ t } \right\rangle \\
=& -2 \langle\vh^\star, \va_2\rangle \langle\vx^\star, \vx\rangle \langle\vx^\star, \va_1 \rangle.
\end{aligned}
\end{equation*}

Since $\nabla_{\vh^\top} \text{grad}_{\vx} f(\vx,\vh)[\va_1,\va_2] \!\!=\!\!\text{grad}_{\vx^\top} \nabla_{\vh}f(\vx,\vh)[\va_2,\va_1]$, for $\va = [\va_1^\top \ \va_2^\top]^\top$, the bilinear form of the Riemannian Hessian at a critical point is
\begin{align*}
&\text{Hess}\, f(\vx,\vh)[\va,\va]\\
=& \text{Hess}_{\vx\vx^\top} f(\vx,\vh)[\va_1,\va_1] +\nabla_{\vh\vh^\top}^2 f(\vx,\vh)[\va_2,\va_2]  + 2 \nabla_{\vh^\top} \text{grad}_{\vx} f(\vx,\vh)[\va_1,\va_2] \\
=& 2 \|\vh\|_2^2 \|\va_1\|_2^2   -   2\langle\vh^\star ,  \vh\rangle  \langle\vx^\star ,  \va_1\rangle^2 + \|\va_2\|_2^2  -  4 \langle\vh^\star ,  \va_2\rangle  \langle\vx^\star ,  \vx\rangle  \langle\vx^\star , \va_1\rangle.
\end{align*}

\section[Proof of Lemma 4.1]{Proof of Lemma~\ref{left ortho upper bound for CP factors real main paper}}
\label{Technical tools used in proofs}

We now restate and prove Lemma 4.1.
\begin{Lemma}
\label{left ortho upper bound for CP factors real_repeated}
Let $\tT = \vx \circ \vx \circ \vh$ and $\tT^\star = \vx^\star \circ {\vx^\star} \circ \vh^\star$ with $\vx,~ \vx^\star\in\RR^N$, $\|\vx\|_2 = \|\vx^\star\|_2 = 1$, and $\vh,~ \vh^\star\in\RR^K$. Assume $\|\vh\|_2\leq \frac{3\|\vh^\star\|_2}{2} = \frac{3\|\tT^\star\|_F}{2}$. Then, 
\begin{align}
    \label{bound of two distances real main paper_repeated}
   \frac{2}{3}\|\tT - \tT^\star\|_F^2 \leq \text{dist}^2(\vx,\vh) \leq 6\|\tT - \tT^\star\|_F^2,
\end{align}
where $\text{dist}^2(\vx,\vh) = \min_{a\in\pm 1}2\|\tT^\star\|_F^2\|\vx - a\vx^\star \|_2^2  +  \|\vh - \vh^\star\|_2^2$.
\end{Lemma}

{
\begin{proof}
Since $\|\vx\|_2=\|\vx^\star\|_2=1$, the factor distance satisfies
\begin{equation}
\label{eq:factor_distance_identity}
\begin{aligned}
\operatorname{dist}^2(\vx,\vh)
&=4\bigl(1-|\langle\vx,\vx^\star\rangle|\bigr)
  \|\vh^\star\|_2^2
  +\|\vh-\vh^\star\|_2^2.
\end{aligned}
\end{equation}
For the lower bound, direct expansion and the Cauchy--Schwarz inequality give
\begin{align*}
\|\tT-\tT^\star\|_F^2
=&\|\vh-\vh^\star\|_2^2
 +2\bigl(1-|\langle\vx,\vx^\star\rangle|^2\bigr)
   \langle\vh,\vh^\star\rangle\\
\leq&\|\vh-\vh^\star\|_2^2
 +3\bigl(1-|\langle\vx,\vx^\star\rangle|^2\bigr)
   \|\vh^\star\|_2^2\\
\leq&\|\vh-\vh^\star\|_2^2
 +6\bigl(1-|\langle\vx,\vx^\star\rangle|\bigr)
   \|\vh^\star\|_2^2\\
\leq&\frac{3}{2}\operatorname{dist}^2(\vx,\vh),
\end{align*}
where we used $\|\vh\|_2\leq 3\|\vh^\star\|_2/2$ and
$1-|\langle\vx,\vx^\star\rangle|^2
\leq 2\bigl(1-|\langle\vx,\vx^\star\rangle|\bigr)$.
This implies the first inequality
in~\eqref{bound of two distances real main paper_repeated}.

For the upper bound, direct expansion also gives
\[
\begin{aligned}
\|\tT-\tT^\star\|_F^2
&=\bigl\|\vh-|\langle\vx,\vx^\star\rangle|^2
  \vh^\star\bigr\|_2^2+\bigl(1-|\langle\vx,\vx^\star\rangle|^4\bigr)
  \|\vh^\star\|_2^2.
\end{aligned}
\]
Note that
\[
\vh-\vh^\star
=\bigl(\vh-|\langle\vx,\vx^\star\rangle|^2\vh^\star\bigr)
-\bigl(1-|\langle\vx,\vx^\star\rangle|^2\bigr)\vh^\star.
\]
The squared triangle inequality then yields
\[
\begin{aligned}
\|\vh-\vh^\star\|_2^2
&\leq 2\bigl\|\vh-|\langle\vx,\vx^\star\rangle|^2
\vh^\star\bigr\|_2^2+2\bigl(1-|\langle\vx,\vx^\star\rangle|^2\bigr)^2
\|\vh^\star\|_2^2.
\end{aligned}
\]
Substituting this estimate into~\eqref{eq:factor_distance_identity} gives
\begin{align*}
\operatorname{dist}^2(\vx,\vh)&\leq 2\bigl\|\vh-|\langle\vx,\vx^\star\rangle|^2
  \vh^\star\bigr\|_2^2+\Bigl[4\bigl(1-|\langle\vx,\vx^\star\rangle|\bigr)
 +2\bigl(1-|\langle\vx,\vx^\star\rangle|^2\bigr)^2\Bigr]
 \|\vh^\star\|_2^2\\
&\leq 2\bigl\|\vh-|\langle\vx,\vx^\star\rangle|^2
  \vh^\star\bigr\|_2^2
 +6\bigl(1-|\langle\vx,\vx^\star\rangle|^4\bigr)
  \|\vh^\star\|_2^2\\
&\leq 6\|\tT-\tT^\star\|_F^2,
\end{align*}
where we used
$1-|\langle\vx,\vx^\star\rangle|
\leq 1-|\langle\vx,\vx^\star\rangle|^4$ and
$\bigl(1-|\langle\vx,\vx^\star\rangle|^2\bigr)^2
\leq 1-|\langle\vx,\vx^\star\rangle|^4$.
This establishes the upper bound.
\end{proof}
}


\section[Proof of Theorem 4.2]{Proof of Theorem~\ref{Local Convergence of Riemannian in the CP PCA_Theorem real}}
\label{Proof of local convergence in the tensor factorization}

\begin{proof}

Recall that the updates of $\vx$ and $\vh$ in RGD are given by:
\begin{align*}
    \vx_{t+1} &= \text{Retr}_{\vx}(\vx_{t} - \frac{\mu}{2\|\tT^\star\|_F^2} \calP_{\text{T}_{\vx} \text{St}}(\nabla_{\vx}f(\vx_t,\vh_t))),\\
    \label{GD_3 orth CP PCA real}
    \vh_{t+1} &= \vh_{t} - \mu \nabla_{\vh}f(\vx_t,\vh_t),
\end{align*}
where $\text{Retr}_{\vx}(\cdot)$ is the standard normalization retraction on the unit sphere, and $\calP_{\text{T}_{\vx} \text{St}}(\cdot ) = (\mId - \vx\vx^\top)(\cdot)$ denotes the orthogonal projection onto the tangent space
of the unit sphere at $\vx$.

The Euclidean gradients of $f(\vx,\vh)$ with respect to $\vx$ and $\vh$ evaluated at the current updates $(\vx_t,\vh_t)$ are given by
\begin{align*}
\nabla_{\vx}f(\vx_t,\vh_t ) =&\underbrace{(\calM_1(\tT_t -\tT^\star))(\vh_t\otimes \vx_t)}_{\vb_1} +\underbrace{(\calM_2(\tT_t - \tT^\star))(\vx_t \otimes \vh_t)}_{\vb_2},\\
\nabla_{\vh}f(\vx_t,\vh_t)=& (\calM_3(\tT_t - \tT^\star))(\vx_t\otimes \vx_t),
\end{align*}
where $\tT_t = \vx_t \circ \vx_t \circ \vh_t$.


We begin by assuming that the iterates remain within a local region, namely
\begin{eqnarray}
    \label{assumption of distance}
    \text{dist}^2(\vx_{t},\vh_{t})\leq \frac{\|\tT^\star \|_F^2}{288},
\end{eqnarray}
which is satisfied at initialization ($t=0$) and will be rigorously established for all $t \geq 1$ via induction.
Under this assumption, we can derive the following bound on $\|\vh_t\|_2^2$:
\begin{equation}
\label{bound proof of h term in the factorization}
\begin{aligned}
\|\vh_t\|_2
&\leq \|\vh^\star\|_2+\|\vh_t-\vh^\star\|_2\leq \left(1+\frac{1}{\sqrt{288}}\right)\|\tT^\star\|_F,\\
\|\vh_t\|_2^2
&\leq \left(1+\frac{1}{\sqrt{288}}\right)^2\|\tT^\star\|_F^2
<\frac{9}{8}\|\tT^\star\|_F^2,
\end{aligned}
\end{equation}
which implies that $\|\vh_t\|_2^2$ remains uniformly bounded for all iterates within the region.

For each $t$, choose
\(
a_t\in\argmin_{a\in\{-1,1\}}\|\vx_t-a\vx^\star\|_2^2.
\)
Throughout the following estimates, $a_t$ denotes this fixed alignment sign.
We measure progress in terms of the factor distance:
\begin{equation}
\label{distance recursion start}
\begin{aligned}
\text{dist}^2\!(\vx_{t+1},\vh_{t+1})
=&\min_{a\in\{-1,1\}}2\|\tT^\star\|_F^2\|\vx_{t+1} \!-\! a\vx^\star \|_2^2 
    +  \|\vh_{t+1} \!-\! \vh^\star\|_2^2 \\
\leq& 2\|\tT^\star\|_F^2
    \| \vx_{t} 
    \!- \!\frac{\mu}{2\|\tT^\star\|_F^2} \!
    \calP_{\text{T}_{\vx} \text{St}}(\nabla_{\vx}f(\vx_t,\!\vh_t)) \!-\!  a_{t}\vx^\star \|_2^2\\
& +  \|\vh_{t} - \mu\nabla_{\vh}f(\vx_t,\vh_t) - \vh^\star\|_2^2,
\end{aligned}
\end{equation}
where the inequality follows by evaluating the minimum at $a_t$ and using the non-expansiveness of the retraction operator~\cite[Lemma 1]{li2021weakly}.
Using the decomposition $\nabla_{\vx}f(\vx_t,\vh_t) = \vb_1 + \vb_2$, we can further expand \eqref{distance recursion start} into
\begin{equation}
\label{distance of factors in orth CP PCA derivation real}
\begin{aligned}
&\text{dist}^2(\vx_{t+1},\vh_{t+1}) \\
\leq& \|\tT^\star\|_F^2\bigg(\|
    \vx_{t} \!-\! \frac{\mu}{\|\tT^\star\|_F^2} 
    \calP_{\text{T}_{\vx} \text{St}}(\vb_1) \!-\! a_{t}\vx^\star\! 
    \|_2^2 
 + \|\vx_{t} \!- \!\frac{\mu}{\|\tT^\star\|_F^2} 
    \calP_{\text{T}_{\vx} \text{St}}(\vb_2)\! -\! a_{t}\vx^\star \|_2^2\bigg) \\
    &+ \|\vh_{t}\! -\! \mu\nabla_{\vh}f(\vx_t,\vh_t) \!-\! \vh^\star\|_2^2 \\
=& \text{dist}^2(\vx_{t},\vh_{t}) 
    \!+\! \mu^2\!\bigg(\!\!\frac{1}{\|\tT^\star\|_F^2}\|\calP_{\!\text{T}_{\vx} \text{St}}(\vb_1)\|_2^2 
   \! +\! \frac{1}{\|\tT^\star\!\|_F^2}\|\calP_{\!\text{T}_{\vx} \text{St}}(\vb_2)\|_2^2  \!+ \!\|\nabla_{\vh}f(\vx_t,\vh_t)\|_2^2\!\!\bigg) \\
    &\!-\! 2\mu\bigg(\!\!\langle\vx_{t}\!-\! a_{t}\vx^\star,  \calP_{\text{T}_{\vx} \text{St}}(\vb_1) \rangle 
 \!+\! \langle\vx_{t} \!-\! a_{t}{\vx^\star}, \calP_{\text{T}_{\vx} \text{St}}(\vb_2)  \rangle \!+ \!\langle\vh_{t}\!-\! \vh^\star,   \nabla_{\vh}f(\vx_t,\vh_t) \rangle\!\! \bigg)\!.
\end{aligned}
\end{equation}


By~\eqref{bound proof of h term in the factorization} and standard norm inequalities, we have
\begin{equation}
\begin{aligned}
\|\calP_{\text{T}_{\vx} \text{St}}(\vb_1)\|_2 &\leq\|\vb_1\|_2 
\leq \frac{3\|\tT^\star\|_F}{2\sqrt{2}} \|\tT_t - \tT^\star\|_F, 
\\
\|\calP_{\text{T}_{\vx} \text{St}}(\vb_2)\|_2&\leq\|\vb_2\|_2 
\leq \frac{3\|\tT^\star\|_F}{2\sqrt{2}} \|\tT_t - \tT^\star\|_F, 
\\
\|\nabla_{\vh}f(\vx_t,\vh_t)\|_2 &\leq
\|\tT_t - \tT^\star\|_F.
\label{squared terms in orth CP PCA derivation real}
\end{aligned}
\end{equation}


Combining the above bounds, we obtain
\begin{equation}
\begin{aligned}
\frac{1}{\|\tT^\star\|_F^2}\|\calP_{\text{T}_{\vx} \text{St}}(\vb_1)\|_2^2
  \!+\! \frac{1}{\|\tT^\star\|_F^2}\|\calP_{\text{T}_{\vx} \text{St}}(\vb_2)\|_2^2 
  \!+\! \|\nabla_{\vh}f(\vx_t,\vh_t)\|_2^2 
  \!\leq\!\! \frac{13}{4}\|\tT_t \!-\! \tT^\star \|_F^2,
\label{squared terms in orth CP PCA upper bound real}
\end{aligned}
\end{equation}
which provides a uniform control of the quadratic terms appearing in the descent relation. In particular, it shows that the squared norms of the projected gradient components and the update in $\vh$ can be bounded in terms of the current tensor error $\|\tT_t - \tT^\star\|_F^2$. We will use this estimate in the subsequent step to establish the contraction of the distance metric $\text{dist}^2(\vx_t,\vh_t)$.

To bound the third part in equation~\eqref{distance of factors in orth CP PCA derivation real}, we first analyze the unprojected cross term, leading to inequality~\eqref{expansion of cross term in the orth PCA real}. Since this estimate alone does not fully capture the entire cross-term contribution, we complement it with a separate analysis of the orthogonal component in equation~\eqref{orth complement cross term real}. In particular, we have
\begin{equation}
    \begin{aligned}
\label{expansion of cross term in the orth PCA real}
&\langle\vx_{t}- a_{t}\vx^\star,  \vb_1 \rangle
   + \langle\vx_{t} - a_{t}{\vx^\star}, \vb_2  \rangle
   + \langle\vh_{t}- \vh^\star,   \nabla_{\vh}f(\vx_t,\vh_t) \rangle  \\
=& \langle (\vx_t\!-\! a_{t}\vx^\star)\! \circ \!\vx_t\! \circ\! \vh_t, \!\tT_t \!-\! \tT^\star  \rangle \!+\! \langle \vx_t \!\circ\! (\vx_t \!-\! a_{t}{\vx^\star}) \!\circ\! \vh_t, \!\tT_t \!-\! \tT^\star  \rangle  \\
& + \langle \vx_t \circ \vx_t  \circ (\vh_t - \vh^\star), \tT_t - \tT^\star  \rangle  \\
=& \langle \tT_t - \tT^\star,\;
        \tT_t - \tT^\star 
        + (\vx_t- a_{t}\vx^\star) \circ (\vx_t- a_{t}{\vx^\star}) \circ\vh_t  \\
&  + \!(\vx_t\!-\! a_{t}\vx^\star) \circ a_{t}{\vx^\star} \!\circ\! (\vh_t \!-\! \vh^\star)\!  
        +\! \vx_t \!\circ\! (\vx_t\!-\! a_{t}{\vx^\star}) \!\circ\! (\vh_t \!-\! \vh^\star)
      \rangle  \\
\geq& \|\tT_t \!-\! \tT^\star \|_F^2 
      \!-\! \frac{1}{2} \|\tT_t \!-\! \tT^\star \|_F^2\!-\! \frac{1}{2}\big\| (\vx_t\!-\! a_{t}\vx^\star) \!\circ\! (\vx_t\!-\! a_{t}{\vx^\star}) \!\circ\!\vh_t  \\
&+ \!(\vx_t\!-\! a_{t}\vx^\star) \!\circ\! a_{t}{\vx^\star}\! \circ\! (\vh_t \!-\! \vh^\star) \!+\! \vx_t \!\circ\! (\vx_t\!-\! a_{t}{\vx^\star}) \!\circ\! (\vh_t \!-\! \vh^\star)
      \big\|_F^2  \\
\geq& \frac{1}{2} \|\tT_t - \tT^\star \|_F^2
      - \frac{3}{2} \big(
        \|(\vx_t- a_{t}\vx^\star) \circ (\vx_t- a_{t}{\vx^\star}) \circ\vh_t\|_F^2  \\
&+\! \|\vx_t \!\circ\! (\vx_t- a_{t}{\vx^\star}) \!\circ \!(\vh_t \!-\! \vh^\star)\|_F^2  
+\! \|(\vx_t\!-\! a_{t}\vx^\star) \!\circ\! a_{t}{\vx^\star} \!\circ\! (\vh_t \!- \!\vh^\star)\|_F^2
        \big)  \\
\geq& \frac{1}{2} \|\tT_t - \tT^\star \|_F^2 
      - \frac{9}{16\|\tT^\star\|_F^2} \text{dist}^4(\vx_{t},\vh_{t}),
\end{aligned}
\end{equation}
where the second equality follows from~\cite[Lemma 14]{qin2024guaranteed} and the last inequality follows from the following estimate:
\begin{equation}
    \begin{aligned}
&\|(\vx_t\!-\! a_{t}\vx^\star) \!\circ\! (\vx_t\!-\! a_{t}{\vx^\star})\! \circ\!\vh_t\|_F^2  
+ \|\vx_t \!\circ\! (\vx_t\!-\! a_{t}{\vx^\star}) \!\circ\! (\vh_t \!-\! \vh^\star)\|_F^2  \\
& \!+\! \|(\vx_t\!-\! a_{t}\vx^\star) \!\circ\! a_{t}{\vx^\star}\! \circ \!(\vh_t \!- \!\vh^\star)\|_F^2  \\
\leq& \|\vx_t\!-\! a_{t}\vx^\star\|_2^2\|\vx_t\!-\! a_{t}{\vx^\star}\|_2^2 \|\vh_t\|_2^2 
+ \|\vx_t\!-\! a_{t}\vx^\star \|_2^2 \| a_{t}{\vx^\star}\|_2^2 \|\vh_t \!-\! \vh^\star \|_2^2  \\
& + \|\vx_t \|_2^2\|\vx_t- a_{t}{\vx^\star} \|_2^2 \|\vh_t - \vh^\star\|_2^2  \\
\leq& \frac{9\|\tT^\star\|_F^2}{8}\|\vx_t- a_{t}\vx^\star\|_2^4 
   + 2\|\vx_t- a_{t}\vx^\star \|_2^2 \|\vh_t - \vh^\star \|_2^2  \\
\leq& \frac{3}{8\|\tT^\star\|_F^2} \text{dist}^4(\vx_{t},\vh_{t}),
\label{expansion of cross term in the orth PCA2 real}
\end{aligned}
\end{equation}
where the second inequality uses~\eqref{bound proof of h term in the factorization} and the last inequality follows from the definition of $\text{dist}$ and
\begin{align*}
&\frac{3}{8\|\tT^\star\|_F^2}\text{dist}^4(\vx_t,\vh_t)
-\frac{9}{8}\|\tT^\star\|_F^2\|\vx_t-a_t\vx^\star\|_2^4
-2\|\vx_t-a_t\vx^\star\|_2^2\|\vh_t-\vh^\star\|_2^2\\
=&\frac{3}{8\|\tT^\star\|_F^2}
\left(\|\tT^\star\|_F^2\|\vx_t-a_t\vx^\star\|_2^2
-\|\vh_t-\vh^\star\|_2^2\right)^2
+\frac{1}{4}\|\vx_t-a_t\vx^\star\|_2^2\|\vh_t-\vh^\star\|_2^2\\
\geq &0.
\end{align*}

Since the unit sphere is a special case of the Stiefel manifold, we can invoke the general formula for the orthogonal complement projection (See~\cite[eqn. (81)]{qin2024guaranteed}). This yields
\begin{eqnarray}
    \label{orth complement term1 real}
    \calP_{\text{T}_{\vx} \text{St}}^\perp(\vx_{t}- a_{t}\vx^\star ) &\!\!\!\!=\!\!\!\!& \frac{1}{2}\vx_{t}((\vx_{t} - a_{t}\vx^\star)^\top(\vx_{t} - a_{t}\vx^\star)).
\end{eqnarray}

Next, we establish an upper bound for the contribution of the orthogonal complement component, which complements the unprojected cross-term estimate obtained earlier. In particular, we derive
\begin{equation}
    \begin{aligned}
    \label{orth complement cross term real}
    &\<\calP_{\text{T}_{\vx} \text{St}}^\perp(\vx_{t}- a_{t}\vx^\star ), \vb_1 \> + \<\calP_{\text{T}_{\vx} \text{St}}^\perp(\vx_{t}- a_{t}{\vx^\star} ), \vb_2 \>\\
     =& \frac{1}{2}\<\vx_{t}((\vx_{t} \!-\! a_{t}\vx^\star)^\top(\vx_{t} \!-\! a_t\vx^\star)),(\calM_1(\tT_t \!-\! \tT^\star))(\vh_t\!\otimes\! \vx_t) \> \\ 
     &+\frac{1}{2}\<\vx_{t}((\vx_{t} \!-\! a_{t}\vx^\star)^\top(\vx_{t} \!-\! a_t\vx^\star)),  (\calM_2(\tT_t \!-\! \tT^\star))(\vx_t\!\otimes\!\vh_t)\>\\
     \leq & \frac{3\|\tT^\star\|_F}{2\sqrt{2}} \|\vx_{t} - a_{t}\vx^\star\|_2^2\|\tT_t - \tT^\star \|_F\\
     \leq& \frac{1}{4}\|\tT_t-\tT^\star\|_F^2
+\frac{9}{8}\|\tT^\star\|_F^2\|\vx_t-a_t\vx^\star\|_2^4\\
\leq& \frac{1}{4}\|\tT_t-\tT^\star\|_F^2
+\frac{9}{32\|\tT^\star\|_F^2}\text{dist}^4(\vx_t,\vh_t).
\end{aligned}
\end{equation}

Combining \eqref{expansion of cross term in the orth PCA real} and \eqref{orth complement cross term real}, we  obtain
\begin{align}
&
   \langle\vx_{t}- a_{t}\vx^\star,  \calP_{\text{T}_{\vx} \text{St}}(\vb_1) \rangle 
   + \langle\vx_{t} - a_{t}{\vx^\star}, \calP_{\text{T}_{\vx} \text{St}}(\vb_2)  \rangle \notag \\
& + \langle\vh_{t}- \vh^\star,   \nabla_{\vh}f(\vx_t,\vh_t) \rangle \notag \\
=& \langle\vx_{t}\!-\! a_{t}\vx^\star,  \vb_1 \rangle 
   \!+\! \langle\vx_{t} \!-\! a_{t}{\vx^\star}, \vb_2  \rangle  \!+\! \langle\vh_{t}\!-\! \vh^\star,   \nabla_{\vh}f(\vx_t,\vh_t) \rangle \notag \\
&-  \langle\calP_{\text{T}_{\vx} \text{St}}^\perp(\vx_{t}- a_{t}\vx^\star ), \vb_1 \rangle - \langle\calP_{\text{T}_{\vx} \text{St}}^\perp(\vx_{t}- a_{t}{\vx^\star} ), \vb_2 \rangle \notag \\
\geq& \frac{1}{4}\|\tT_t-\tT^\star\|_F^2
-\frac{27}{32\|\tT^\star\|_F^2}\text{dist}^4(\vx_t,\vh_t)\notag\\
\geq& \frac{1}{4}\|\tT_t-\tT^\star\|_F^2
-\frac{1}{\|\tT^\star\|_F^2}\text{dist}^4(\vx_t,\vh_t)\notag\\
\geq& \frac{13}{96}\|\tT_t-\tT^\star\|_F^2
+\frac{1}{64}\text{dist}^2(\vx_t,\vh_t),
\label{expansion of cross term in the orth PCA final real}
\end{align}
where the last line follows from Lemma~\ref{left ortho upper bound for CP factors real_repeated} and the induction hypothesis~\eqref{assumption of distance}, which together give
\begin{align*}
\frac{11}{96}\|\tT_t\!-\!\tT^\star\|_F^2
\!-\!\frac{\text{dist}^4(\vx_t,\vh_t)}{\|\tT^\star\|_F^2}
\!\geq\!\left(\!\frac{11}{576}-\frac{1}{288}\!\right)\!
\text{dist}^2(\vx_t,\vh_t)
\!=\!\frac{1}{64}\text{dist}^2(\vx_t,\vh_t).
\end{align*}

Finally, combining inequalities~\eqref{squared terms in orth CP PCA upper bound real} and \eqref{expansion of cross term in the orth PCA final real}, we  obtain
\begin{equation}
    \begin{aligned}
    \label{distance of factors in orth CP PCA derivation final real}
    \operatorname{dist}^2(\vx_{t+1},\vh_{t+1})
&\leq
\left(1-\frac{\mu}{32}\right)
\operatorname{dist}^2(\vx_t,\vh_t)+
\left(\frac{13\mu^2}{4}-\frac{13\mu}{48}\right)
\|\tT_t-\tT^\star\|_F^2\\
&\leq
\left(1-\frac{\mu}{32}\right)
\operatorname{dist}^2(\vx_t,\vh_t),
\end{aligned}
\end{equation}
provided that $\mu\leq \frac{1}{12}$.

\paragraph*{Proof of \eqref{assumption of distance} by induction} 
First note that \eqref{assumption of distance} holds at $t = 0$ by initialization. Suppose it holds at $t = t'$, so that $\|\vh_{t'}\|_2^2\leq \frac{9\|\tT^\star\|_F^2}{8}$. By invoking \eqref{distance of factors in orth CP PCA derivation final real}, we then have $\text{dist}^2(\vx_{t'+1},\vh_{t'+1}) \le \text{dist}^2(\vx_{t'},\vh_{t'})$. Hence, \eqref{assumption of distance} also holds at $ t= t'+1$. By induction, we can conclude that \eqref{assumption of distance} holds for all $t\ge 0$, thereby completing the proof.

\end{proof}

\section[Proof of Theorem 5.2]{Proof of Theorem~\ref{TENSOR SENSING SPECTRAL INITIALIZATION CP}}
\label{Proof of spectral initialization in the noiseless}

\begin{proof}
We aim to bound the reconstruction error
$\|\tT_0-\tT^\star\|_F$ of the spectral
initialization in (12).
Let $\tE=\frac{1}{m}\calA^*\calA(\tT^\star)-\tT^\star$. It follows from (12) that
\[
\vh_0=\langle\vx_0,\vx^\star\rangle^2\vh^\star+\calM_3(\tE)(\vx_0\otimes\vx_0).
\]
Substituting this identity into $\tT_0=\vx_0\circ\vx_0\circ\vh_0$ gives
\begin{equation}
\label{eq:spectral_init_error_decomposition}
\tT_0-\tT^\star
=\vx_0\circ\vx_0\circ\bigl[\calM_3(\tE)(\vx_0\otimes\vx_0)\bigr]
+\bigl[\langle\vx_0,\vx^\star\rangle^2\vx_0\circ\vx_0\circ\vh^\star-\tT^\star\bigr].
\end{equation}
Since $\|\vx_0\|_2=1$, for any $\vh\in\RR^K$,
\[
\bigl\langle\langle\vx_0,\vx^\star\rangle^2\vx_0\circ\vx_0\circ\vh^\star-\tT^\star,\,\vx_0\circ\vx_0\circ\vh\bigr\rangle=0.
\]
The first term in~\eqref{eq:spectral_init_error_decomposition} has the form $\vx_0\circ\vx_0\circ\vh$, so the two terms are orthogonal. Taking squared Frobenius norms and using $\|\vx_0\|_2=1$ gives
\begin{equation}
\label{eq:spectral_init_squared_error}
\|\tT_0-\tT^\star\|_F^2
=\|\calM_3(\tE)(\vx_0\otimes\vx_0)\|_2^2
+\|\tT^\star\|_F^2\bigl(1-\langle\vx_0,\vx^\star\rangle^4\bigr).
\end{equation}
It remains to bound the two terms on the right-hand side of~\eqref{eq:spectral_init_squared_error}.

Let $\eta=\|\calM_1(\tE)\|_2$. For the first term,
\[
\begin{aligned}
\|\calM_3(\tE)(\vx_0\!\otimes\!\vx_0)\|_2
\!=\!\sup_{\|\vz\|_2=1}\!|\vz^\top\calM_3(\tE)(\vx_0\!\otimes\!\vx_0)|
\!=\!\sup_{\|\vz\|_2=1}\!|\vx_0^\top\calM_1(\tE)(\vz\!\otimes\!\vx_0)|\!\le\!\eta.
\end{aligned}
\]
For the second term, since $\vx_0$ is a leading left singular vector of $\calM_1(\tT^\star+\tE)$, $\vx_0\vx_0^\top\calM_1(\tT^\star+\tE)$ is a best rank-one approximation to $\calM_1(\tT^\star+\tE)$. As $\calM_1(\tT^\star)$ has rank at most one, we have
\[
\|(\mathbf I-\vx_0\vx_0^\top)\calM_1(\tT^\star+\tE)\|_2\le\|\calM_1(\tT^\star+\tE)-\calM_1(\tT^\star)\|_2=\eta.
\]
It follows that
\[
\begin{aligned}
&\|\tT^\star\|_F\sqrt{1-\langle\vx_0,\vx^\star\rangle^2}
=\|(\mathbf I-\vx_0\vx_0^\top)\calM_1(\tT^\star)\|_2\\
\le& \|(\mathbf I-\vx_0\vx_0^\top)\calM_1(\tT^\star+\tE)\|_2+\|(\mathbf I-\vx_0\vx_0^\top)\calM_1(\tE)\|_2\le2\eta.
\end{aligned}
\]
Then, we have
\begin{equation}
\label{eq:spectral_init_error_eta}
\|\tT_0-\tT^\star\|_F^2
\le\eta^2+2\|\tT^\star\|_F^2\bigl(1-\langle\vx_0,\vx^\star\rangle^2\bigr)
\le9\eta^2,
\end{equation}
where the first inequality follows from $1-\langle\vx_0,\vx^\star\rangle^4\le2\bigl(1-\langle\vx_0,\vx^\star\rangle^2\bigr)$.

It remains to bound $\eta$ using the TRIP.
Let $\mathbf E_k=\tE(:,:,k)$. Since $\calM_1(\tE)=[\mathbf E_1\ \cdots\ \mathbf E_K]$, we have
\begin{equation}
\label{eq:spectral_init_eta_slices}
\eta^2=\left\|\sum_{k=1}^K\mathbf E_k\mathbf E_k^\top\right\|_2\le\sum_{k=1}^K\|\mathbf E_k\|_2^2.
\end{equation}
It therefore suffices to show that $\|\mathbf E_k\|_2\le\delta_r\|\tT^\star\|_F$ for every $k$.
For each \(k\), the definition of the spectral norm gives
\[
\|\mathbf E_k\|_2=\sup_{\|\va\|_2=\|\vb\|_2=1}|\va^\top\mathbf E_k\vb|.
\]
Let $\ve_k$ be the $k$-th standard basis vector of $\RR^K$. For any normalized vectors $\va,\vb\in\RR^N$, the tensor $\va\circ\vb\circ\ve_k$ has CP rank one and unit Frobenius norm. By the definition of the slice, we have
\(
\va^\top\mathbf E_k\vb
=\langle\tE,\va\circ\vb\circ\ve_k\rangle.
\)
When \(\tT^\star=0\), we have \(\tT_0=0\), so the bound in Theorem~\ref{TENSOR SENSING SPECTRAL INITIALIZATION CP} holds. For \(\tT^\star\ne0\), define two tensors
\[
\tZ_+=\frac{\tT^\star}{\|\tT^\star\|_F}+\va\circ\vb\circ\ve_k,\quad \tZ_-=\frac{\tT^\star}{\|\tT^\star\|_F}-\va\circ\vb\circ\ve_k.
\]
Both have CP rank at most two, so the assumed rank-two TRIP gives
\[
\left|\frac1m\|\calA(\tZ_\pm)\|_2^2
-\|\tZ_\pm\|_F^2\right|
\le\delta_r\|\tZ_\pm\|_F^2,
\]
which further implies that
\[
\begin{aligned}
|\va^\top\mathbf E_k\vb|
&=|\langle\tE,\va\circ\vb\circ\ve_k\rangle|=\left|\frac1m\langle\calA(\tT^\star),\calA(\va\circ\vb\circ\ve_k)\rangle-\langle\tT^\star,\va\circ\vb\circ\ve_k\rangle\right|\\
&=\frac{\|\tT^\star\|_F}{4}\left|\frac1m\bigl(\|\calA(\tZ_+)\|_2^2-\|\calA(\tZ_-)\|_2^2\bigr)-\|\tZ_+\|_F^2+\|\tZ_-\|_F^2\right|\\
&\le\frac{\delta_r\|\tT^\star\|_F}{4}\bigl(\|\tZ_+\|_F^2+\|\tZ_-\|_F^2\bigr)=\delta_r\|\tT^\star\|_F,
\end{aligned}
\]
where the last equality follows from $\|\tZ_+\|_F^2+\|\tZ_-\|_F^2=4$.
Taking the supremum over all normalized vectors $\va,\vb$ therefore gives
\(
\|\mathbf E_k\|_2\le\delta_r\|\tT^\star\|_F.
\)
Substituting this bound into~\eqref{eq:spectral_init_eta_slices} yields
\(
\eta^2\le K\delta_r^2\|\tT^\star\|_F^2.
\)
Together with~\eqref{eq:spectral_init_error_eta}, this proves Theorem~\ref{TENSOR SENSING SPECTRAL INITIALIZATION CP}.
\end{proof}

\section[Proof of Theorem 5.3]{Proof of Theorem~\ref{Local Convergence of Riemannian in the CP sensing_Theorem real}} 
\label{Proof of local convergence of CP sensing}

We first state a useful property of the TRIP: inner products between CP-format tensors are approximately preserved under the sensing operator.
\begin{Lemma} 
\label{RIP CONDITION FRO THE CP SENSING OTHER PROPERTY}
Suppose $\calA$ satisfies the TRIP with constant $\delta_{r}$ for $r=2$.
If every tensor in $\operatorname{span}\{\tX_1,\tX_2\}$ has CP rank at most two, then
\begin{align}
\label{eq:rip_result1}
    \bigg|\bigg\<\bigg(\frac{1}{m}\calA^*\calA-\mathcal{I}\bigg)(\tX_1), \tX_2\bigg\>\bigg|\leq \delta_{r}\|\tX_1\|_F\|\tX_2\|_F.
\end{align}
\end{Lemma}

\begin{proof}
The claim holds when either tensor is zero. Note that
\[
\begin{aligned}
\langle\tX_1,\tX_2\rangle
&=\frac14\bigl(\|\tX_1+\tX_2\|_F^2-\|\tX_1-\tX_2\|_F^2\bigr),\\
\langle\calA(\tX_1),\calA(\tX_2)\rangle
&=\frac14\bigl(\|\calA(\tX_1+\tX_2)\|_2^2-\|\calA(\tX_1-\tX_2)\|_2^2\bigr).
\end{aligned}
\]
When both tensors are nonzero, the left-hand side of~\eqref{eq:rip_result1} can be written as
\[
\begin{aligned}
&\left|\frac1m\langle\calA(\tX_1),\calA(\tX_2)\rangle-\langle\tX_1,\tX_2\rangle\right|\\
=&\frac{\|\tX_1\|_F\|\tX_2\|_F}{4}\Bigg|
\left(\frac1m\!\left\|\calA\left(\frac{\tX_1}{\|\tX_1\|_F}+\frac{\tX_2}{\|\tX_2\|_F}\right)\right\|_2^2
-\left\|\frac{\tX_1}{\|\tX_1\|_F}+\frac{\tX_2}{\|\tX_2\|_F}\right\|_F^2\right)\\
&-
\left(\frac1m\!\left\|\calA\left(\frac{\tX_1}{\|\tX_1\|_F}-\frac{\tX_2}{\|\tX_2\|_F}\right)\right\|_2^2
-\left\|\frac{\tX_1}{\|\tX_1\|_F}-\frac{\tX_2}{\|\tX_2\|_F}\right\|_F^2\right)
\Bigg|.
\end{aligned}
\]
Applying the triangle inequality and the TRIP, we can upper bound this error by
\[
\begin{aligned}
&
\frac{\delta_r\|\tX_1\|_F\|\tX_2\|_F}{4}
\left(\left\|\frac{\tX_1}{\|\tX_1\|_F}+\frac{\tX_2}{\|\tX_2\|_F}\right\|_F^2+
\left\|\frac{\tX_1}{\|\tX_1\|_F}-\frac{\tX_2}{\|\tX_2\|_F}\right\|_F^2\right)\\
&=\delta_r\|\tX_1\|_F\|\tX_2\|_F,
\end{aligned}
\]
where the two squared Frobenius norms sum to $4$.
\end{proof}

\begin{proof}[Proof of Theorem~\ref{Local Convergence of Riemannian in the CP sensing_Theorem real}]

The Euclidean gradients of $g(\vx,\vh)$ with respect to $\vx$ and $\vh$ evaluated at the current updates $(\vx_t,\vh_t)$ are given by
\begin{align*}
\nabla_{\vx}g(\vx_t,\vh_t)
=& \vc_1 + \vc_2, \notag \\
\nabla_{\vh}g(\vx_t,\!\vh_t)
=& \frac{1}{m}\sum_{i=1}^{m} \big(\big\langle \tA_i, \vx_t \circ \vx_t \circ \vh_t\!\big\rangle \!-\!\! \vy(i) \big) \times \calM_3(\tA_i)(\vx_t\otimes \vx_t)
\end{align*}
with
\begin{align*}
&\vc_1 = \frac{1}{m}\sum_{i=1}^{m} \big(\big\langle \tA_i, \vx_t \circ \vx_t \circ \vh_t\big\rangle - \vy(i) \big) \times \calM_1(\tA_i)(\vh_t\otimes \vx_t), \\
&\vc_2 = \frac{1}{m}\sum_{i=1}^{m} \big(\big\langle \tA_i, \vx_t \circ \vx_t \circ \vh_t\big\rangle - \vy(i) \big) \times \calM_2(\tA_i)(\vx_t\otimes\vh_t).
\label{gradient 3 orth CP sensing real}
\end{align*}

Write the RGD update as
\[
\vv=\vx_t-\frac{\mu}{2\|\tT^\star\|_F^2}\operatorname{grad}_{\vx}g(\vx_t,\vh_t),
\qquad
\vx_{t+1}=\frac{\vv}{\|\vv\|_2}.
\]
Since $\langle\vx_t,\operatorname{grad}_{\vx}g(\vx_t,\vh_t)\rangle=0$ and $\|\vx_t\|_2=1$, we have
\[
\begin{aligned}
\|\vv\|_2^2=1+\frac{\mu^2}{4\|\tT^\star\|_F^4}\|\operatorname{grad}_{\vx}g(\vx_t,\vh_t)\|_2^2\geq1.
\end{aligned}
\]
For each $t$, choose $a_t\in\argmin_{a\in\{-1,1\}}\|\vx_t-a\vx^\star\|_2^2$. It follows that 
\[
\begin{aligned}
\|\vv-a_t\vx^\star\|_2^2-\left\|\frac{\vv}{\|\vv\|_2}-a_t\vx^\star\right\|_2^2=(\|\vv\|_2-1)\left(\|\vv\|_2+1-2\left\langle\frac{\vv}{\|\vv\|_2},a_t\vx^\star\right\rangle\right)\geq0,
\end{aligned}
\]
which further implies
\[
\|\vx_{t+1}-a_t\vx^\star\|_2^2\leq\|\vv-a_t\vx^\star\|_2^2.
\]
Now, we can expand the distance metric in iteration $t+1$
as
\begin{equation}
\label{distance of factors in orth CP sensing derivation real}
\begin{aligned}
&\operatorname{dist}^2(\vx_{t+1},\vh_{t+1})=\min_{a\in\{-1,1\}}2\|\tT^\star\|_F^2\|\vx_{t+1}-a\vx^\star\|_2^2+\|\vh_{t+1}-\vh^\star\|_2^2\\
\leq&2\|\tT^\star\|_F^2\|\vx_{t+1}-a_t\vx^\star\|_2^2+\|\vh_{t+1}-\vh^\star\|_2^2\\
\leq&2\|\tT^\star\|_F^2\left\|\vx_t-\frac{\mu}{2\|\tT^\star\|_F^2}\operatorname{grad}_{\vx}g(\vx_t,\vh_t)-a_t\vx^\star\right\|_2^2\\
&+\|\vh_t-\mu\nabla_{\vh}g(\vx_t,\vh_t)-\vh^\star\|_2^2\\
=&\operatorname{dist}^2(\vx_t,\vh_t) +\mu^2\left(\frac{\|\operatorname{grad}_{\vx}g(\vx_t,\vh_t)\|_2^2}{2\|\tT^\star\|_F^2}+\|\nabla_{\vh}g(\vx_t,\vh_t)\|_2^2\right)\\
&-2\mu\Bigl(\langle\vx_t-a_t\vx^\star,\operatorname{grad}_{\vx}g(\vx_t,\vh_t)\rangle +\langle\vh_t-\vh^\star,\nabla_{\vh}g(\vx_t,\vh_t)\rangle\Bigr).
\end{aligned}
\end{equation}

{
For the squared gradient terms in~\eqref{distance of factors in orth CP sensing derivation real}, we have
\begin{equation}
\label{eq:cp_sensing_gradient_square_reduction}
\begin{aligned}
&\frac{\|\operatorname{grad}_{\vx}g(\vx_t,\vh_t)\|_2^2}{2\|\tT^\star\|_F^2}+\|\nabla_{\vh}g(\vx_t,\vh_t)\|_2^2\leq\frac{\|\vc_1\|_2^2+\|\vc_2\|_2^2}{\|\tT^\star\|_F^2}+\|\nabla_{\vh}g(\vx_t,\vh_t)\|_2^2,
\end{aligned}
\end{equation}
where we used $\|\operatorname{grad}_{\vx}g(\vx_t,\vh_t)\|_2^2\leq\|\vc_1+\vc_2\|_2^2\leq2(\|\vc_1\|_2^2+\|\vc_2\|_2^2)$. Thus, we need bounds on $\|\vc_1\|_2$, $\|\vc_2\|_2$, and $\|\nabla_{\vh}g(\vx_t,\vh_t)\|_2$.
Recall $\tT_t=\vx_t\circ\vx_t\circ\vh_t$. As in Appendix~\ref{Proof of local convergence in the tensor factorization}, $\vb_1$ and $\vb_2$ denote the corresponding population components, obtained by replacing $\calA^*\calA/m$ by the identity operator.
To obtain these bounds, the triangle inequality gives
\[
\begin{aligned}
&\|\vc_i\|_2\leq\|\vb_i\|_2+\|\vc_i-\vb_i\|_2,\quad i=1,2,\\
&\|\nabla_{\vh}g(\vx_t,\vh_t)\|_2\leq\|\nabla_{\vh}f(\vx_t,\vh_t)\|_2+\|\nabla_{\vh}g(\vx_t,\vh_t)-\nabla_{\vh}f(\vx_t,\vh_t)\|_2.
\end{aligned}
\]
For the population gradients, we have
\[
\|\nabla_{\vh}f(\vx_t,\vh_t)\|_2\leq\|\tT_t-\tT^\star\|_F, \quad \|\vb_i\|_2\leq\|\tT_t-\tT^\star\|_F\|\vh_t\|_2,\quad i=1,2,
\]
where we used the last inequality in~\eqref{squared terms in orth CP PCA derivation real}, the Cauchy--Schwarz inequality, and $\|\vx_t\|_2=1$. 

It remains to control the differences between the sensing and population gradients using the TRIP.
By the dual characterization of the vector norm, we have
\begin{equation}
\label{eq:cp_sensing_first_gradient_error}
\begin{aligned}
\|\vc_1-\vb_1\|_2
&=\sup_{\|\va\|_2=1}\left|
\left\langle\left(\frac1m\calA^*\calA-\mathcal I\right)(\tT_t-\tT^\star),\va\circ\vx_t\circ\vh_t\right\rangle\right|\\
&\leq\delta_r\|\tT_t-\tT^\star\|_F\|\vh_t\|_2,
\end{aligned}
\end{equation}
where the inequality follows from Lemma~\ref{RIP CONDITION FRO THE CP SENSING OTHER PROPERTY} with $r=2$. Indeed, for arbitrary $\alpha,\beta\in\RR$, $\va\in\RR^N$, and $\vb\in\RR^K$,
\[
\begin{aligned}
\alpha(\tT_t-\tT^\star)+\beta\va\circ\vx_t\circ\vh_t
&=(\alpha\vx_t+\beta\va)\circ\vx_t\circ\vh_t-\alpha\tT^\star,\\
\alpha(\tT_t-\tT^\star)+\beta\vx_t\circ\va\circ\vh_t
&=\vx_t\circ(\alpha\vx_t+\beta\va)\circ\vh_t-\alpha\tT^\star,\\
\alpha(\tT_t-\tT^\star)+\beta\vx_t\circ\vx_t\circ\vb
&=\vx_t\circ\vx_t\circ(\alpha\vh_t+\beta\vb)-\alpha\tT^\star.
\end{aligned}
\]
Each right-hand side has CP rank at most two, so the lemma applies to all three pairings.
}
{
The other two pairings similarly give
}

\begin{equation}
\label{eq:cp_sensing_gradient_differences}
\begin{aligned}
&\|\vb_2-\vc_2\|_2\leq
\delta_r 
\|\tT_t-\tT^\star\|_F {\|\vh_t\|_2},\\
&\|\nabla_{\vh}f(\vx_t,\vh_t)-\nabla_{\vh}g(\vx_t,\vh_t)\|_2
\leq\delta_r\|\tT_t-\tT^\star\|_F.
\end{aligned}
\end{equation}
Consequently, we obtain
\begin{equation}
\label{squared term  in the tensor sensing 11 to 31}
\begin{aligned}
&\|\vc_i\|_2\leq(1+\delta_r)\|\tT_t-\tT^\star\|_F{\|\vh_t\|_2},\quad i=1,2,\\
&\|\nabla_{\vh}g(\vx_t,\vh_t)\|_2\leq
(1+\delta_r)\|\tT_t-\tT^\star\|_F.
\end{aligned}
\end{equation}
{
Substituting~\eqref{squared term  in the tensor sensing 11 to 31} into~\eqref{eq:cp_sensing_gradient_square_reduction} gives
\begin{equation}
\label{eq:cp_sensing_gradient_square_intermediate}
\begin{aligned}
\frac{\|\operatorname{grad}_{\vx}g(\vx_t,\vh_t)\|_2^2}{2\|\tT^\star\|_F^2}+\|\nabla_{\vh}g(\vx_t,\vh_t)\|_2^2 \leq(1+\delta_r)^2\left(\frac{2\|\vh_t\|_2^2}{\|\tT^\star\|_F^2}+1\right)\|\tT_t-\tT^\star\|_F^2.
\end{aligned}
\end{equation}

It remains to bound $\|\vh_t\|_2$ and relate $\|\tT_t-\tT^\star\|_F$ to the factor distance.
The choice of $a_t$ gives $\langle\vx_t,a_t\vx^\star\rangle\geq0$, and
\[
\operatorname{dist}^2(\vx_t,\vh_t)
=4\|\tT^\star\|_F^2\bigl(1-\langle\vx_t,a_t\vx^\star\rangle\bigr)+\|\vh_t-\vh^\star\|_2^2.
\]
We assume temporarily that
\begin{equation}
\label{bound requirement of noiseless tensor sensing}
\operatorname{dist}^2(\vx_t,\vh_t)\leq\frac{\|\tT^\star\|_F^2}{400},
\end{equation}
and verify by induction at the end that this holds for every iterate.
It follows from~\eqref{bound requirement of noiseless tensor sensing} that
\begin{equation}
\label{eq:cp_sensing_local_factor_bounds}
\|\vh_t-\vh^\star\|_2\leq\frac{\|\tT^\star\|_F}{20},\quad
\|\vh_t\|_2\leq\frac{21}{20}\|\tT^\star\|_F,\quad
\langle\vx_t,a_t\vx^\star\rangle\geq1-\frac1{1600}.
\end{equation}
Then, we get
\begin{equation}
\label{eq:cp_sensing_tensor_factor_upper}
\begin{aligned}
\|\tT_t-\tT^\star\|_F^2
&=\|\vh_t-\vh^\star\|_2^2+
2\bigl(1-\langle\vx_t,a_t\vx^\star\rangle^2\bigr)\langle\vh_t,\vh^\star\rangle\\
&\leq\|\vh_t-\vh^\star\|_2^2+
\frac{21}{5}\|\tT^\star\|_F^2\bigl(1-\langle\vx_t,a_t\vx^\star\rangle\bigr)\\
&\leq\frac{21}{20}\operatorname{dist}^2(\vx_t,\vh_t),
\end{aligned}
\end{equation}
where the first inequality uses $\langle\vh_t,\vh^\star\rangle\leq21\|\tT^\star\|_F^2/20$ and $1-\langle\vx_t,a_t\vx^\star\rangle^2\leq2(1-\langle\vx_t,a_t\vx^\star\rangle)$. Substituting~\eqref{eq:cp_sensing_local_factor_bounds} and~\eqref{eq:cp_sensing_tensor_factor_upper} into~\eqref{eq:cp_sensing_gradient_square_intermediate} gives
\begin{equation}
\label{eq:cp_sensing_quadratic_new}
\begin{aligned}
\frac{\|\operatorname{grad}_{\vx}g(\vx_t,\vh_t)\|_2^2}{2\|\tT^\star\|_F^2}+\|\nabla_{\vh}g(\vx_t,\vh_t)\|_2^2 \leq4(1+\delta_r)^2\operatorname{dist}^2(\vx_t,\vh_t).
\end{aligned}
\end{equation}
}

{
For the gradient inner products in~\eqref{distance of factors in orth CP sensing derivation real}, we have
\[
\begin{aligned}
&\langle\vx_t-a_t\vx^\star,\operatorname{grad}_{\vx}g(\vx_t,\vh_t)\rangle+\langle\vh_t-\vh^\star,\nabla_{\vh}g(\vx_t,\vh_t)\rangle\\
\geq&\langle\vx_t-a_t\vx^\star,\operatorname{grad}_{\vx}f(\vx_t,\vh_t)\rangle+\langle\vh_t-\vh^\star,\nabla_{\vh}f(\vx_t,\vh_t)\rangle\\
&-\Bigl|\langle\vx_t-a_t\vx^\star,\operatorname{grad}_{\vx}g(\vx_t,\vh_t)-\operatorname{grad}_{\vx}f(\vx_t,\vh_t)\rangle\\
&+\langle\vh_t-\vh^\star,\nabla_{\vh}g(\vx_t,\vh_t)-\nabla_{\vh}f(\vx_t,\vh_t)\rangle\Bigr|.
\end{aligned}
\]
It therefore suffices to bound the population inner products from below and the absolute gradient-error contribution from above. The gradient formulas for the population objective $f$ give
\[
\begin{aligned}
\operatorname{grad}_{\vx}f(\vx_t,\vh_t)
&=-2\langle\vx_t,a_t\vx^\star\rangle\langle\vh_t,\vh^\star\rangle
\bigl(a_t\vx^\star-\langle\vx_t,a_t\vx^\star\rangle\vx_t\bigr),\\
\nabla_{\vh}f(\vx_t,\vh_t)&=\vh_t-\langle\vx_t,a_t\vx^\star\rangle^2\vh^\star.
\end{aligned}
\]
Consequently,
\begin{equation}
\label{eq:cp_sensing_population_cross_new}
\begin{aligned}
&\langle\vx_t-a_t\vx^\star,\operatorname{grad}_{\vx}f(\vx_t,\vh_t)\rangle+
\langle\vh_t-\vh^\star,\nabla_{\vh}f(\vx_t,\vh_t)\rangle\\
=&\|\vh_t\!-\!\vh^\star\|_2^2\!+\!
\bigl(1\!-\!\langle\vx_t,a_t\vx^\star\rangle^2\bigr)
\Bigl[2\langle\vx_t,a_t\vx^\star\rangle\|\tT^\star\|_F^2\!+\!
\bigl(2\langle\vx_t,a_t\vx^\star\rangle\!+\!1\bigr)
\langle\vh_t\!-\!\vh^\star,\vh^\star\rangle\Bigr]\\
\geq&\|\vh_t-\vh^\star\|_2^2+\|\tT^\star\|_F^2\bigl(1-\langle\vx_t,a_t\vx^\star\rangle\bigr)\bigl(1+\langle\vx_t,a_t\vx^\star\rangle\bigr)\left(\frac{19}{10}\langle\vx_t,a_t\vx^\star\rangle-\frac1{20}\right)\\
\geq&\|\vh_t-\vh^\star\|_2^2+\left(4-\frac6{1600}-\frac3{10}\right)\|\tT^\star\|_F^2\bigl(1-\langle\vx_t,a_t\vx^\star\rangle\bigr)\\
\geq&\|\vh_t-\vh^\star\|_2^2+\frac{18}{5}\|\tT^\star\|_F^2\bigl(1-\langle\vx_t,a_t\vx^\star\rangle\bigr)\\
\geq&\frac9{10}\operatorname{dist}^2(\vx_t,\vh_t),
\end{aligned}
\end{equation}
where the first inequality uses $\langle\vh_t-\vh^\star,\vh^\star\rangle\geq-\|\tT^\star\|_F^2/20$ from~\eqref{eq:cp_sensing_local_factor_bounds}. The second inequality uses $1-1/1600\leq\langle\vx_t,a_t\vx^\star\rangle\leq1$. Writing the product as $2\langle\vx_t,a_t\vx^\star\rangle(1+\langle\vx_t,a_t\vx^\star\rangle)-(1+\langle\vx_t,a_t\vx^\star\rangle)(2\langle\vx_t,a_t\vx^\star\rangle+1)/20$, the first term is at least $4-6/1600$, and the subtracted term is at most $3/10$. The last inequality follows from the definition of $\operatorname{dist}$.
}

{
For the gradient-error contribution, we have
\begin{equation}
\label{eq:cp_sensing_cross_perturbation}
\begin{aligned}
&\Bigl|\!\langle\vx_t\!-\!a_t\vx^\star,\operatorname{grad}_{\vx}g(\vx_t,\vh_t)\!-\!\operatorname{grad}_{\vx}f(\vx_t,\vh_t)\rangle \!+\!\langle\vh_t\!-\!\vh^\star,\nabla_{\vh}g(\vx_t,\vh_t)\!-\!\nabla_{\vh}f(\vx_t,\vh_t)\rangle\!\Bigr|\\
\leq&\delta_r\|\tT_t-\tT^\star\|_F
\left(2\|\vh_t\|_2\sqrt{1-\langle\vx_t,a_t\vx^\star\rangle^2}+\|\vh_t-\vh^\star\|_2\right)\\
\leq&\delta_r\|\tT_t-\tT^\star\|_F
\sqrt{\frac{2\|\vh_t\|_2^2}{\|\tT^\star\|_F^2}+1}
\sqrt{2\|\tT^\star\|_F^2(1-\langle\vx_t,a_t\vx^\star\rangle^2)+\|\vh_t-\vh^\star\|_2^2}\\
\leq&2\delta_r\operatorname{dist}^2(\vx_t,\vh_t),
\end{aligned}
\end{equation}
where the first inequality uses the gradient error bounds in~\eqref{eq:cp_sensing_first_gradient_error} and~\eqref{eq:cp_sensing_gradient_differences} and the self-adjointness of the tangent projection. The tangent projection has operator norm one, and we used
\(
\|\calP_{\mathrm T_{\vx}\mathrm{St}}(\vx_t-a_t\vx^\star)\|_2
=\sqrt{1-\langle\vx_t,a_t\vx^\star\rangle^2}.
\)
The second inequality follows from the Cauchy--Schwarz inequality.
For the last inequality, the final square root is bounded by the factor distance. Equations~\eqref{eq:cp_sensing_local_factor_bounds} and~\eqref{eq:cp_sensing_tensor_factor_upper} give $\frac{21}{20}\bigl(2(21/20)^2+1\bigr)<4$.
Combining~\eqref{eq:cp_sensing_population_cross_new} and~\eqref{eq:cp_sensing_cross_perturbation} with the decomposition above, we obtain
\begin{equation}
\label{eq:cp_sensing_cross_new}
\begin{aligned}
\langle\vx_t-a_t\vx^\star,\operatorname{grad}_{\vx}g(\vx_t,\vh_t)\rangle+\langle\vh_t-\vh^\star,\nabla_{\vh}g(\vx_t,\vh_t)\rangle
\geq\frac{9-20\delta_r}{10}\operatorname{dist}^2(\vx_t,\vh_t).
\end{aligned}
\end{equation}
}

{
Substituting~\eqref{eq:cp_sensing_cross_new} and~\eqref{eq:cp_sensing_quadratic_new} into~\eqref{distance of factors in orth CP sensing derivation real} gives
\begin{equation}
\label{eq:cp_sensing_improved_contraction}
\begin{aligned}
\operatorname{dist}^2(\vx_{t+1},\vh_{t+1})
&\leq\left(1-\frac{9-20\delta_r}{5}\mu+4(1+\delta_r)^2\mu^2\right)
\operatorname{dist}^2(\vx_t,\vh_t)\\
&\leq\left(1-\frac{9-20\delta_r}{10}\mu\right)\operatorname{dist}^2(\vx_t,\vh_t),
\end{aligned}
\end{equation}
where the second inequality follows from $0<\mu\leq(9-20\delta_r)/(40(1+\delta_r)^2)$. The contraction factor lies in $(0,1)$ because $0\leq\delta_r<9/20$ and $(9-20\delta_r)\mu/10\leq81/400<1$. Condition~\eqref{bound requirement of noiseless tensor sensing} holds at $t=0$ by assumption, and~\eqref{eq:cp_sensing_improved_contraction} preserves it at the next iterate. Induction therefore proves both the local-region condition and the stated contraction for every $t\geq0$.
}

\end{proof}

\section[Proof of Theorem 5.4]{Proof of Theorem~\ref{TENSOR SENSING SPECTRAL INITIALIZATION CP noisy}}
\label{Proof of spectral initialization in the noise}

\begin{proof}
We extend the spectral initialization analysis
of Theorem~\ref{TENSOR SENSING SPECTRAL INITIALIZATION CP} to the noisy setting.
By~\eqref{eq:spectral_init_error_eta} in Appendix~\ref{Proof of spectral initialization in the noiseless}, the spectral initializer satisfies the deterministic perturbation bound

\[
\begin{aligned}
\|\tT_0-\tT^\star\|_F
&\leq
3\left\|
\calM_1\left(
\frac{1}{m}\calA^*(\mathbf y)-\tT^\star
\right)
\right\|_2.
\end{aligned}
\]
Substituting
$\mathbf y=\calA(\tT^\star)+\mathbf e$
and applying the triangle inequality,
we obtain
\[
\begin{aligned}
\|\tT_0-\tT^\star\|_F
&\leq
3\left\|
\calM_1\left(
\frac{1}{m}\calA^*\calA(\tT^\star)
-\tT^\star
\right)
\right\|_2+
3\left\|
\calM_1\left(
\frac{1}{m}\calA^*(\mathbf e)
\right)
\right\|_2\\
&\leq
3\sqrt K\,\delta_r
\|\tT^\star\|_F+
3\left\|
\calM_1\left(
\frac{1}{m}\calA^*(\mathbf e)
\right)
\right\|_2,
\end{aligned}
\]
where the last inequality follows from the slice concatenation estimate~\eqref{eq:spectral_init_eta_slices} and the subsequent rank-two TRIP bounds in Appendix~\ref{Proof of spectral initialization in the noiseless}. It remains to control the noise term.

Since
$\mathbf e\sim\mathcal N(\mathbf 0,\gamma^2 I_m)$
is independent of $\calA$, for any unit
vectors $\va,\vp\in\RR^N$ and
$\vq\in\RR^K$, conditional on $\calA$,
\[
\begin{aligned}
&
\frac{1}{m}
\langle\mathbf e,
\calA(\va\circ\vp\circ\vq)\rangle\sim
\mathcal N\left(
0,
\frac{\gamma^2}{m^2}
\|\calA(\va\circ\vp\circ\vq)\|_2^2
\right).
\end{aligned}
\]
Because $\va\circ\vp\circ\vq$ has
CP rank one and unit Frobenius norm,
the TRIP implies
\[
\frac{1}{m}
\|\calA(\va\circ\vp\circ\vq)\|_2^2
\leq 1+\delta_r.
\]
Hence, the Gaussian tail inequality gives
\[
\begin{aligned}
&
\Pr\left(
\left|
\frac{1}{m}
\langle\mathbf e,
\calA(\va\circ\vp\circ\vq)\rangle
\right|>t
\,\middle|\,\calA
\right)\leq
2\exp\left(
-\frac{mt^2}
{2\gamma^2(1+\delta_r)}
\right).
\end{aligned}
\]


Applying~\cite[Theorem~1]{tomioka2014spectral} to the random tensor $\calA^*(\mathbf e)/m$ conditionally on $\calA$, with dimensions $(N,N,K)$, variance proxy $\gamma^2(1+\delta_r)/m$, and failure probability $2\exp(-(2N+K))$, yields the following uniform bound for a universal constant $C>0$:
\begin{equation}
\begin{aligned}
&
\sup_{
\|\va\|_2=\|\vp\|_2=\|\vq\|_2=1}
\left|
\frac{1}{m}
\langle\mathbf e,
\calA(\va\circ\vp\circ\vq)\rangle
\right|
\leq
C\gamma\sqrt{\frac{(1+\delta_r)(2N+K)}{m}}.
\end{aligned}
\label{eq:noisy_uniform_rank_one_bound}
\end{equation}

To bound the matrix noise term, we have
\[
\begin{aligned}
\left\|\calM_1\left(\frac1m\calA^*(\mathbf e)\right)\right\|_2^2
&\leq\sum_{k=1}^K\left\|\left[\frac1m\calA^*(\mathbf e)\right](:,:,k)\right\|_2^2 \leq C^2\gamma^2\frac{K(1+\delta_r)(2N+K)}{m},
\end{aligned}
\]
where the first inequality is the slice concatenation estimate~\eqref{eq:spectral_init_eta_slices} in Appendix~\ref{Proof of spectral initialization in the noiseless}, and the second follows from~\eqref{eq:noisy_uniform_rank_one_bound} by taking $\vq$ to be the $k$-th standard basis vector and taking the supremum over unit vectors $\va,\vp$.
Taking square roots and using $2N+K\leq2(N+K)$ gives
\[
\begin{aligned}
\left\|\calM_1\left(\frac1m\calA^*(\mathbf e)\right)\right\|_2
&\leq C\gamma\sqrt{\frac{K(1+\delta_r)(2N+K)}m}\\
&=O\left(\gamma\sqrt{\frac{K(1+\delta_r)(N+K)}m}\right).
\end{aligned}
\]

Combining the above noise bound with the
initial perturbation estimate yields
\[
\begin{aligned}
\|\tT_0-\tT^\star\|_F
&\leq
3\sqrt{K}\,\delta_r
\|\tT^\star\|_F +
O\left(\gamma\sqrt{\frac{K(1+\delta_r)(N+K)}{m}}\right).
\end{aligned}
\]
This bound holds on the high-probability
event established above, conditional on
$\calA$ satisfying the TRIP.
This completes the proof.

\end{proof}

\section[Proof of Theorem 5.5]{Proof of Theorem~\ref{Local Convergence of Riemannian in the CP sensing_Theorem real noisy}}
\label{Proof of local convergence of CP sensing noisy}

\begin{proof}
Throughout this proof, let $G$ denote the noisy sensing objective and let $g$ denote the corresponding noiseless sensing objective, namely,
\[
\begin{aligned}
G(\vx,\vh)&=\frac{1}{2m}\|\calA(\vx\circ\vx\circ\vh-\tT^\star)-\ve\|_2^2,\\
g(\vx,\vh)&=\frac{1}{2m}\|\calA(\vx\circ\vx\circ\vh-\tT^\star)\|_2^2.
\end{aligned}
\]
The gradients of $G(\vx,\vh)$ at iteration $t$ are given as: 
\begin{align*}
\nabla_{\vx}G(\vx_t,\vh_t)
=& \vf_1 + \vf_2,
\\
\nabla_{\vh}G(\vx_t,\vh_t)
=& \frac{1}{m} \sum_{i=1}^{m}
  \big(\langle\tA_i,\; \vx_t \circ \vx_t \circ \vh_t - \vx^\star \circ {\vx^\star} \circ \vh^\star \rangle - \ve_i \big)  \times \calM_3(\tA_i)(\vx_t\otimes \vx_t),
\end{align*}
where
\begin{align*}
&\vf_1 \!=\! \frac{1}{m} \sum_{i=1}^{m}
  \big(\langle\tA_i,\; \vx_t \circ \vx_t \circ \vh_t \!- \vx^\star \circ {\vx^\star} \circ \vh^\star \rangle \!-\! \ve_i \big)
  \!\times\! \calM_1(\tA_i)(\vh_t\!\otimes\! \vx_t), \\
&\vf_2 \!=\! \frac{1}{m} \sum_{i=1}^{m}
  \big(\langle\tA_i,\; \vx_t \circ \vx_t \circ \vh_t \!- \vx^\star \circ {\vx^\star} \circ \vh^\star \rangle \!-\! \ve_i \big)
  \!\times\! \calM_2(\tA_i)(\vx_t\!\otimes\!\vh_t).
\end{align*}

For each $t$, choose $a_t\in\argmin_{a\in\{-1,1\}}\|\vx_t-a\vx^\star\|_2^2$. We assume temporarily that
\begin{equation}
\label{Bound requirement of tensor sensing noise environment}
\operatorname{dist}^2(\vx_t,\vh_t)\leq\frac{\|\tT^\star\|_F^2}{400},
\end{equation}
and verify this condition by induction at the end. Applying the normalized-update argument preceding~\eqref{distance of factors in orth CP sensing derivation real} to the noisy objective gives
\begin{equation}
\label{distance of factors in orth CP sensing derivation real noisy}
\begin{aligned}
&\operatorname{dist}^2(\vx_{t+1},\vh_{t+1})\\
\leq{}&\operatorname{dist}^2(\vx_t,\vh_t)
+\mu^2\left(\frac{\|\operatorname{grad}_{\vx}G(\vx_t,\vh_t)\|_2^2}{2\|\tT^\star\|_F^2}+\|\nabla_{\vh}G(\vx_t,\vh_t)\|_2^2\right)\\
&-2\mu\Bigl(\langle\vx_t-a_t\vx^\star,\operatorname{grad}_{\vx}G(\vx_t,\vh_t)\rangle
+\langle\vh_t-\vh^\star,\nabla_{\vh}G(\vx_t,\vh_t)\rangle\Bigr).
\end{aligned}
\end{equation}

For the squared gradient terms in~\eqref{distance of factors in orth CP sensing derivation real noisy}, the inequality $\|\va+\vb\|_2^2\leq2\|\va\|_2^2+2\|\vb\|_2^2$ gives
\begin{equation}
\label{eq:noisy_gradient_square_reduction}
\begin{aligned}
&\frac{\|\operatorname{grad}_{\vx}G(\vx_t,\vh_t)\|_2^2}{2\|\tT^\star\|_F^2}+\|\nabla_{\vh}G(\vx_t,\vh_t)\|_2^2\\
\leq{}&2\left(\frac{\|\operatorname{grad}_{\vx}g(\vx_t,\vh_t)\|_2^2}{2\|\tT^\star\|_F^2}+\|\nabla_{\vh}g(\vx_t,\vh_t)\|_2^2\right)\\
&+2\Biggl(\frac{\|\operatorname{grad}_{\vx}G(\vx_t,\vh_t)-\operatorname{grad}_{\vx}g(\vx_t,\vh_t)\|_2^2}{2\|\tT^\star\|_F^2}\\
&\qquad+\|\nabla_{\vh}G(\vx_t,\vh_t)-\nabla_{\vh}g(\vx_t,\vh_t)\|_2^2\Biggr).
\end{aligned}
\end{equation}
The first parenthesis is controlled by~\eqref{eq:cp_sensing_quadratic_new}. It remains to bound the noise-gradient differences in the second parenthesis.

We work on the uniform noise event in~\eqref{eq:noisy_uniform_rank_one_bound}, which has probability at least $1-2\exp(-(2N+K))$ conditional on $\calA$ satisfying the TRIP.
On the uniform noise event, define
\[
\xi=\sup_{\|\va\|_2=\|\vb\|_2=\|\vz\|_2=1}
\left|\frac{1}{m}\langle\ve,\calA(\va\circ\vb\circ\vz)\rangle\right|
\leq C\gamma\sqrt{\frac{N+K}{m}},
\]
where $\va,\vb\in\RR^N$, $\vz\in\RR^K$, and $C>0$ is a universal constant; we used $\delta_r<2/5$ to absorb the factor $1+\delta_r$. The difference between the two objectives consists only of the noise term. By the definition of $\xi$ and the dual characterization of the vector norm, we have
\[
\begin{aligned}
\|\operatorname{grad}_{\vx}G(\vx_t,\vh_t)-\operatorname{grad}_{\vx}g(\vx_t,\vh_t)\|_2&\leq2\xi\|\vh_t\|_2,\\
\|\nabla_{\vh}G(\vx_t,\vh_t)-\nabla_{\vh}g(\vx_t,\vh_t)\|_2&\leq\xi,
\end{aligned}
\]
where the first bound also uses the fact that the tangent projection has operator norm one. Hence,
\begin{equation}
\label{eq:noisy_gradient_difference_weighted}
\begin{aligned}
&\frac{\|\operatorname{grad}_{\vx}G(\vx_t,\vh_t)-\operatorname{grad}_{\vx}g(\vx_t,\vh_t)\|_2^2}{2\|\tT^\star\|_F^2}
+\|\nabla_{\vh}G(\vx_t,\vh_t)-\nabla_{\vh}g(\vx_t,\vh_t)\|_2^2\\
&\qquad\leq\left(\frac{2\|\vh_t\|_2^2}{\|\tT^\star\|_F^2}+1\right)\xi^2\leq4\xi^2,
\end{aligned}
\end{equation}
where the last inequality follows from $\|\vh_t\|_2\leq21\|\tT^\star\|_F/20$ in~\eqref{eq:cp_sensing_local_factor_bounds}, which holds under~\eqref{Bound requirement of tensor sensing noise environment}.
Substituting~\eqref{eq:cp_sensing_quadratic_new} and~\eqref{eq:noisy_gradient_difference_weighted} into~\eqref{eq:noisy_gradient_square_reduction} yields
\begin{equation}
\label{squared term  in the noisy tensor sensing final one}
\frac{\|\operatorname{grad}_{\vx}G(\vx_t,\vh_t)\|_2^2}{2\|\tT^\star\|_F^2}+\|\nabla_{\vh}G(\vx_t,\vh_t)\|_2^2
\leq8(1+\delta_r)^2\operatorname{dist}^2(\vx_t,\vh_t)+8\xi^2.
\end{equation}

For the gradient inner products, we similarly have
\begin{equation}
\label{cross term  in the noisy tensor sensing2}
\begin{aligned}
&\langle\vx_t-a_t\vx^\star,\operatorname{grad}_{\vx}G(\vx_t,\vh_t)\rangle
+\langle\vh_t-\vh^\star,\nabla_{\vh}G(\vx_t,\vh_t)\rangle\\
\geq{}&\langle\vx_t-a_t\vx^\star,\operatorname{grad}_{\vx}g(\vx_t,\vh_t)\rangle
+\langle\vh_t-\vh^\star,\nabla_{\vh}g(\vx_t,\vh_t)\rangle
-2\xi\operatorname{dist}(\vx_t,\vh_t)\\
\geq{}&\frac{9-20\delta_r}{10}\operatorname{dist}^2(\vx_t,\vh_t)
-2\xi\operatorname{dist}(\vx_t,\vh_t)\\
\geq{}&\frac{8-20\delta_r}{10}\operatorname{dist}^2(\vx_t,\vh_t)-10\xi^2,
\end{aligned}
\end{equation}
where the first inequality follows from the Cauchy--Schwarz inequality,~\eqref{eq:noisy_gradient_difference_weighted}, and the definition of the factor distance. The second inequality uses~\eqref{eq:cp_sensing_cross_new}. The last inequality uses $2\xi\operatorname{dist}(\vx_t,\vh_t)\leq\operatorname{dist}^2(\vx_t,\vh_t)/10+10\xi^2$.

Substituting~\eqref{squared term  in the noisy tensor sensing final one} and~\eqref{cross term  in the noisy tensor sensing2} into~\eqref{distance of factors in orth CP sensing derivation real noisy} yields
\begin{equation}
\label{distance of factors in orth CP sensing derivation real noisy final}
\begin{aligned}
&\operatorname{dist}^2(\vx_{t+1},\vh_{t+1})\\
\leq{}&\left(1-\frac{8-20\delta_r}{5}\mu+8(1+\delta_r)^2\mu^2\right)\operatorname{dist}^2(\vx_t,\vh_t)
+(20\mu+8\mu^2)\xi^2\\
\leq{}&\left(1-\frac{2(2-5\delta_r)}{5}\mu\right)\operatorname{dist}^2(\vx_t,\vh_t)
+(20\mu+8\mu^2)\xi^2,
\end{aligned}
\end{equation}
where the second inequality follows from $0<\mu\leq(2-5\delta_r)/(20(1+\delta_r)^2)$. The contraction factor lies in $(0,1)$ because $0\leq\delta_r<2/5$ and $\mu\leq1/10$.

Assuming that~\eqref{Bound requirement of tensor sensing noise environment} holds up to iteration~$t$,  iterating~\eqref{distance of factors in orth CP sensing derivation real noisy final} gives
\[
\begin{aligned}
&\operatorname{dist}^2(\vx_{t+1},\vh_{t+1})
\leq{}\left(1-\frac{2(2-5\delta_r)}{5}\mu\right)^{t+1}\operatorname{dist}^2(\vx_0,\vh_0)
+\left[1-\left(1-\frac{2(2-5\delta_r)}{5}\mu\right)^{t+1}\right]\frac{(50+20\mu)\xi^2}{2-5\delta_r}.
\end{aligned}
\]
Since $\mu\leq1/10$, the last fraction is at most $52\xi^2/(2-5\delta_r)$. Under the stated measurement condition, with a sufficiently large universal constant, this is at most $\|\tT^\star\|_F^2/400$. The displayed upper bound is therefore a convex combination of two quantities no greater than $\|\tT^\star\|_F^2/400$. Starting from the assumed initialization condition, induction proves~\eqref{Bound requirement of tensor sensing noise environment} for every $t\geq0$. Finally, substituting the bound on $\xi$ into the above inequality yields the stated convergence estimate, valid simultaneously for all iterations on the uniform noise event.
\end{proof}

\end{document}